\documentclass[letterpaper, 10 pt, conference]{ieeeconf}  

\IEEEoverridecommandlockouts                              

\usepackage{amsfonts}
\usepackage{amstext}
\usepackage{mathtools}
\usepackage{cases}
\usepackage{enumerate}
\usepackage{subfigure}
\usepackage{amssymb}
\usepackage{float}
\usepackage{comment}
\usepackage[dvipsnames]{xcolor}
\usepackage{hyperref}

\newtheorem{thm}{\textbf{\text{Theorem}}}

\newtheorem{pro}{\textbf{\text{Proposition}}}

\newtheorem{pb}{\textbf{\text{Problem}}}

\newtheorem{rmk}{\textbf{\text{Remark}}}

\newcommand{\ie}{\textit{i.e.}}
\newcommand{\eg}{\textit{e.g.}}

\title{\LARGE \bf
Attitude Synchronization on $SO(3)$ for Heterogeneous Multi-Agent Systems Using Vector Measurements}
\author{Mouaad Boughellaba, Soulaimane Berkane, and Abdelhamid Tayebi
\thanks{This work was supported by the National Sciences and Engineering Research Council of Canada (NSERC), under the grants NSERC-DG RGPIN 2020-06270 and NSERC-DG RGPIN-2020-04759, and by Fonds de recherche du Qu\'ebec (FRQ).} 
\thanks{M. Boughellaba and A. Tayebi are with the Department of Electrical Engineering, Lakehead University, Thunder Bay, ON P7B 5E1, Canada \tt\small \{mboughel,atayebi\}@lakeheadu.ca.} 
\thanks{S. Berkane is with the Department of Computer Science and Engineering, University of Quebec in Outaouais, Gatineau, QC, Canada. {\tt\small Soulaimane.Berkane@uqo.ca}}
}

\begin{document}

\maketitle
\thispagestyle{empty}
\pagestyle{empty}

\begin{abstract}
 This paper addresses the distributed attitude synchronization problem for a network of rigid-body systems on the special orthogonal group $SO(3)$. Each agent measures, in its body frame, its own angular velocity and a set of vectors whose corresponding directions in the inertial frame are unknown. Under an undirected, connected, and acyclic interaction graph topology, we develop four distributed synchronization schemes relying solely on local vector measurements, without the need for attitude estimation and attitude exchange between agents. Specifically, two leaderless schemes are proposed at the kinematic and dynamic levels to achieve synchronization to a common unknown orientation. In addition, two leader–follower schemes are proposed to align all agents with a prescribed constant orientation defined by reference vector measurements available only to a designated leader. All control laws are formulated directly on $SO(3)$, preserving the geometric structure of the attitude dynamics. A rigorous stability analysis is provided showing that the closed-loop systems achieve almost global asymptotic stability, which is the strongest stability property one can achieve on $SO(3)$ with smooth controllers. 
Numerical simulations are provided to illustrate the effectiveness and performance of the proposed distributed control schemes.
\end{abstract}

\section{Introduction}
Attitude synchronization is a fundamental problem in the coordination and control of multi-agent rigid-body systems. This problem is central in a wide range of applications, including spacecraft formation flying, cooperative exploration and manipulation. Efficient and robust attitude synchronization is crucial for ensuring coordinated maneuvering and optimal task execution in these multi-agent systems, especially in scenarios where agents operate in decentralized environments with limited communication capabilities. The challenge of attitude synchronization arises from several factors: the nonlinear nature of the attitude kinematics and dynamics, the restrictions imposed by the interaction graph topology, and the need for control strategies that are directly designed on the rotation manifold, specifically the special orthogonal group \( SO(3) \). These complexities make the design of distributed control algorithms that ensure strong stability guarantees and robustness a nontrivial task.

Traditional approaches to attitude synchronization often rely on parameterized representations of rotations, such as Euler angles (EA), Modified Rodrigues Parameters (MRP), and unit quaternions. Early works based on EA and MRP \cite{Dimos_SCL_2009,Bayezit_TIE_2013,Ren_TCAT_2010,Xin_automatic_2020,Ziyang_automatic_2010,Chen_TAES2019} provide convenient minimal representations evolving in the Euclidean space $\mathbb{R}^3$. However, these parameterizations are not homomorphic to the rotation group $SO(3)$ and suffer from singularities, resulting in only local synchronization guarantees and issues such as gimbal lock in certain configurations \cite{Ren_TCAT_2010}. To overcome these limitations, several studies have adopted the unit-quaternion representation \cite{Ren_IJACSP2007,BAI20083170,Liu_automatic2018,Pedro_TAC2020,SAVINO2020142,Zhang_TCS2022}, which provides a global and nonsingular description of rigid-body orientation \cite{Shuster1993ASO}. Quaternion-based synchronization schemes have also been developed using virtual dynamics techniques to eliminate the need for angular velocity measurements \cite{Tayebi_TAC2008,Abdessameud_TAC2009,Abdessameud_TAC2012}. Nevertheless, the unit-quaternion space constitutes a double cover of $SO(3)$, which introduces redundancy and may lead to the undesirable unwinding phenomenon if not properly addressed. To mitigate this issue, quaternion-based control strategies have been developed, achieving global asymptotic stability while preventing unwinding through the use of logic-based switching mechanisms \cite{Mayhew_TAC2011,GUI2018225,Huang_automatic2021}. Despite these developments, the inherent limitations associated with parameterized representations motivate the adoption of geometric control and synchronization methods that operate directly on the rotation manifold $SO(3)$, thereby avoiding singularities and redundancy \cite{Shuster1993ASO}. 
Moreover, These methods enable the development of distributed attitude synchronization schemes that preserve the structure of the rotation manifold, ensuring both mathematical rigor and strong stability properties. Several such schemes on \( SO(3) \) have been proposed in the literature \cite{Maadani_ACC2020,Tron_CDC2012,Tron_TAC2012,Markdahl_TAC2020,SARLETTE2009572,SARLETTE20072232,Alain_SIAM,Wei_TAC2018,Maadani_2022}. Recent work by the authors in \cite{Mouaad_ACC24,Mouaad_TAC2026} proposed attitude synchronization schemes, both with and without angular velocity measurements, directly on $SO(3)$, ensuring global asymptotic stability. In \cite{li2026}, the authors extended the continuous attitude synchronization design in \cite{Mouaad_TAC2026} to tackle the leader-follower attitude synchronization problem. However, most of these approaches rely on full state exchange, where agents exchange either relative or absolute orientations, a requirement that can be challenging to fulfill in practical settings due to the lack of low-cost sensing solutions that provide attitude measurements. To address this, the implementation of these schemes often necessitates the use of an attitude observer, which adds complexity to their deployment.

The attitude synchronization problem under partial state exchange where the neighboring agents share vector measurements rather than full attitude information, is an interesting topic that has not been comprehensively addressed in the literature. The concept of using vector measurements directly in the control, without estimating the attitude, has been exploited in single-agent systems such as in \cite{Tayebi_TAC2013,Tayebi_TAC_2018}. However, the extension of this concept to multi-agent attitude synchronization remains limited in the literature. The works in \cite{Lageman_CDC2009,Sarlette_SIAM2012} developed leaderless attitude synchronization schemes based on unit inter-agent vector measurements. Similarly, \cite{Tran_ACC22} proposed leader-follower synchronization schemes that also use inter-agent vector measurements, yet these approaches primarily focus on the kinematic level of the rotation manifold. Building upon these foundational works, \cite{Thakur_CDC2015} proposed synchronization schemes at the dynamic-level, based on vector measurements, considering both leaderless and leader-follower structures. However, these studies either establish local asymptotic stability or boundedness and convergence results, but fail to provide strong asymptotic stability guarantees such as global or almost-global asymptotic stability.

To address the aforementioned gap, this paper investigates the leaderless and leader–follower attitude synchronization problem for a network of rigid-body systems evolving on $SO(3)$ under an undirected, connected, and acyclic graph topology. Each agent measures its own angular velocity and has access to local measurements of inertial unit-vectors expressed in its body frame, where the corresponding inertial-frame directions are unknown. The contributions of this work can be summarized as follows:
\begin{enumerate}
\item We propose two distributed \textit{leaderless} attitude synchronization schemes formulated at the kinematic and dynamic levels, endowed with almost global asymptotic stability guarantees. Both schemes are designed directly on $SO(3)$ and rely only on local vector measurements, without requiring full-state exchange between neighboring agents, while ensuring that the orientations of all agents converge to a common unknown orientation from almost all initial conditions.

\item We propose two distributed \textit{leader–follower} synchronization schemes, at the kinematic and dynamic levels, endowed with almost global asymptotic stability guarantees. These schemes allow to align all agents orientations with a prescribed constant orientation, defined through reference vector measurements, available to a single agent (the leader). The proposed designs preserve the distributed architecture and the same local measurement requirements as in the leaderless case.

\item For all proposed schemes, we provide a rigorous stability analysis demonstrating that the closed-loop systems enjoy almost global asymptotic stability, the strongest stability result one can achieve with smooth feedback on $SO(3)$.

\item We demonstrate the practical advantage of using local vector measurements in attitude synchronization, which allows for distributed attitude synchronization to a common orientation without requiring global attitude information. Notably, the measurement vectors are unknown in the inertial frame, yet the proposed schemes still enable effective synchronization based on local vector measurements alone.

\end{enumerate}

To the best of the authors knowledge this is the first work in the literature providing almost global asymptotic stability results for the leaderless and leader–follower attitude synchronization problems under the considered measurement model and communication graph topology.

A preliminary version of this work was presented in \cite{Mouaad_CDC2025}. However, \cite{Mouaad_CDC2025} did not include the leader–follower attitude synchronization schemes developed at the kinematic and dynamic levels, nor the analysis of the leaderless scheme for synchronization to a common time-varying orientation, which are presented in this paper. Furthermore, in contrast to the conference version \cite{Mouaad_CDC2025}, the present paper provides complete stability analyses for the leaderless synchronization schemes at both the kinematic and dynamic levels. The paper is also expanded with additional discussions, critical remarks, and more illustrative simulation results.

The remainder of this paper is organized as follows. Section~\ref{s2} presents the mathematical preliminaries and notation used throughout the paper. Section~\ref{s3} formulates the attitude synchronization problem and states the main objectives. The leaderless distributed attitude synchronization schemes are developed and analyzed in Section~\ref{leaderless_section}. Section~\ref{leader-follower_section} presents the leader–follower distributed synchronization schemes along with their stability analysis. Numerical simulations illustrating the effectiveness and performance of the proposed schemes are provided in Section~\ref{s6}. Finally, conclusions and directions for future work are given in Section~\ref{s7}.

\section{Preliminaries}\label{s2}

\subsection{Notations}
\noindent The sets of real numbers and the $n$-dimensional Euclidean space  are denoted by $\mathbb{R}$ and $\mathbb{R}^n$, respectively. The set of unit vectors in $\mathbb{R}^n$ is defined as $\mathbb{S}^{n-1}:=\{x\in \mathbb{R}^n~|~x^\top  x =1\}$. Given two matrices $A$,$B$ $\in \mathbb{R}^{m\times n}$, their Euclidean inner product is defined as $\langle \langle A,B \rangle \rangle=\text{tr}(A^\top  B)$. The Euclidean norm of a vector $x \in \mathbb{R}^n$ is defined as $||x||=\sqrt{x^\top  x}$. The matrix $I_n \in \mathbb{R}^{n \times n}$ denotes the identity matrix, and $\textbf{1}_n=[1\hdots1]^\top  \in \mathbb{R}^n$. Consider a smooth manifold $\mathcal{Q}$ with $\mathcal{T}_x \mathcal{Q}$ being its tangent space at point $x \in \mathcal{Q}$. Let $f: \mathcal{Q} \rightarrow \mathbb{R}_{\geq 0}$ be a continuously differentiable real-valued function. The function $f$ is a potential function on $\mathcal{Q}$ with respect to set $\mathcal{B} \subset \mathcal{Q}$ if $f(x)=0$, $\forall x \in \mathcal{B}$, and $f(x) > 0$, $\forall x\notin \mathcal{B}$. The gradient of $f$ at $x \in \mathcal{Q}$, denoted by $\nabla_x f(x)$, is defined as the unique element of $\mathcal{T}_x \mathcal{Q}$ such that $\dot f(x)=\langle \nabla_x f(x), \eta\rangle_x$, $\forall \eta \in \mathcal{T}_x \mathcal{Q}$, where $\langle~,~\rangle_x:\mathcal{T}_x \mathcal{Q} \times \mathcal{T}_x \mathcal{Q} \rightarrow \mathbb{R}$ is Riemannian metric on $\mathcal{Q}$ \cite{Mahony_book_OAMM}. The point $x \in \mathcal{Q}$ is said to be a critical point of $f$ if $\nabla_x f(x)=0$. The attitude of a rigid body is represented by a rotation matrix $R$ which belongs to the special orthogonal group $SO(3):= \{ R\in \mathbb{R}^{3\times 3} | \hspace{0.1cm}\text{det}(R)=1, R^\top R=I_3\}$. The $SO(3)$ group has a compact manifold structure and its tangent space is given by $\mathcal{T}_RSO(3):=\{R \hspace{0.1cm}\Omega \hspace{0.2cm} | \hspace{0.2cm} \Omega \in \mathfrak{so}(3)\}$ where $\mathfrak{so}(3):=\{ \Omega \in \mathbb{R}^{3\times 3} | \Omega^\top =-\Omega\}$ is the Lie algebra of the matrix Lie group $SO(3)$. The map $[.]^{\times}: \mathbb{R}^3 \rightarrow \mathfrak{so}(3)$ is defined such that $[x]^\times y=x \times y$, for any $x,y \in \mathbb{R}^3$, where $\times$ denotes the vector cross product on $\mathbb{R}^3$. The inverse map of $[.]^{\times}$ is $\text{vex}: \mathfrak{so}(3) \rightarrow \mathbb{R}^3$ such that $\text{vex}([\omega]^\times)=\omega$, and $[\text{vex}(\Omega)]^\times=\Omega$ for all $\omega \in \mathbb{R}^3$ and $\Omega \in \mathfrak{so}(3)$. Also, let $\mathbb{P}_a : \mathbb{R}^{3\times 3} \rightarrow \mathfrak{so}(3)$ be the projection map on the Lie algebra $\mathfrak{so}(3)$ such that $\mathbb{P}_a(A):=(A-A^\top )/2$. Given a 3-by-3 matrix $C:=[c_{ij}]_{i,j=1,2,3}$, one has  $\psi(C) := \text{vex} \circ \mathbb{P}_a (C)=\text{vex}(\mathbb{P}_a(C))=\frac{1}{2}[c_{32}-c_{23},c_{13}-c_{31},c_{21}-c_{12}]^\top $. For any $R\in SO(3)$, the normalized Euclidean distance on $SO(3)$, with respect to the identity $I_3$, is defined as $|R|_I^2:=\frac{1}{4}\text{tr}(I_3-R)$ $\in[0,1]$. The angle-axis parameterization of $SO(3)$, is given by $\mathcal{R}(\theta, v):=I_3+\sin\hspace{0.05cm}\theta \hspace{0.2cm}[v]^\times + (1-\cos\hspace{0.05cm}\theta)([v]^\times)^2$, where $v\in \mathbb{S}^2$ and  $\theta \in \mathbb{R}$ are the rotational axis and angle, respectively.
\subsection{Graph Theory}
Consider a network of $N$ agents. The interaction topology between the agents is described by an undirected (unweighted) graph $\mathcal G = (\mathcal V,\mathcal E)$, where $\mathcal V=\{1,...,N\}$ and $\mathcal E \subseteq \mathcal V \times \mathcal V $ represent the vertex (or agent) set and the edge set of graph $\mathcal{G}$, respectively. In undirected graphs, the edge $(i,j) \in \mathcal E$ indicates that agents $i$ and $j$ interact with each other without any restriction on the direction, which means that agent $i$ can obtain information (via communication, measurements, or both) from agent $j$ and vice versa. The \textit{adjacency} matrix $D = [d_{ij}] \in \mathbb{R}^{N \times N}$ of the graph $\mathcal{G}$ is defined such that $d_{ij} = 1$ if $(i, j) \in \mathcal{E}$ and $d_{ij} = 0$ otherwise. Self-edges are not considered, \ie, $d_{ii} = 0$. The set of neighbors of agent $i$ is defined as $\mathcal N_i = \{j \in \mathcal V : (i,j) \in \mathcal E \}$. The undirected path is a sequence of edges in an undirected graph. An undirected graph is called connected if there is an undirected path between every pair of distinct agents of the graph. An undirected graph has a cycle if there exists an undirected path that starts and ends at the same agent \cite{Ren_book}. An acyclic undirected graph is an undirected graph without a cycle. An undirected tree is an undirected graph in which any two agents are connected by exactly one path (\ie, an undirected tree is an undirected, connected, and acyclic graph). An oriented graph is obtained from an undirected graph by assigning an arbitrary direction to each edge \cite{Mesbahi_book}. Consider an oriented graph where each edge is indexed by a number. Let $M=|\mathcal{E}|$ and $\mathcal{M}=\{1,\hdots, M\}$ be the total number of edges and the set of edge indices, respectively. The \textit{incidence} matrix, denoted by $H\in\mathbb{R}^{N\times M}$, is defined as follows \cite{BAI20083170}:

{\small
\begin{equation}\label{h_matrix}
   H:=[h_{ik}]_{N\times M} \hspace{0.4cm} \text{with} \hspace{0.2cm} h_{ik}=\begin{cases}
      +1 & k\in\mathcal{M}_i^+\\
      -1 & k\in\mathcal{M}_i^-\\
      0 & \text{otherwise}
    \end{cases} \nonumber,    
\end{equation}}where $\mathcal{M}_i^+ \subset \mathcal{M}$ denotes the subset of edge indices in which agent $i$ is the head of the edges and $\mathcal{M}_i^- \subset \mathcal{M}$ denotes the subset of edge indices in which agent $i$ is the tail of the edges. For a connected graph, one verifies that $H^\top \textbf{1}_N=0$ and rank$(H)$=$N-1$. Moreover, the columns of $H$ are linearly independent if the graph is an undirected tree.

\section{Problem Statement}\label{s3}
Consider a network of $N$ agents governed by the following rigid-body rotational dynamics:
\begin{align}
       \dot{R}_i &= R_i[\omega_i]^{\times}\label{R_dynamics_i}\\
       J_i \dot{\omega}_i &= -[\omega_i]^\times J_i \omega_i + \tau_i,\label{w_dynamics_i}
\end{align}
where $R_i \in SO(3)$ represents the orientation of the body-attached frame of agent $i$ with respect to the inertial frame, $\omega_i\in \mathbb{R}^3$ is the body-frame angular velocity of agent $i$, and $\tau_i \in \mathbb{R}^3$ is the control torque to be designed. The matrix $J_i \in \mathbb{R}^{3 \times 3}$ is a constant and known inertia matrix of agent $i$. All agents are assumed to be equipped with identical inertial sensors that measure the same set of inertial unit vectors expressed in their respective body-attached frames, with at least two of these vectors being non-collinear. In addition, each agent is equipped with rate gyros providing body-frame measurements of its angular velocity. The measurements of the inertial vectors available to agent $i$ are given by
\begin{equation}\label{vector_measurement}
    b_\ell^i = R_i^\top a_\ell,
\end{equation}
where $a_\ell \in \mathbb{S}^2$, $\ell = 1,2,\ldots,n$, with $n \geq 2$, represents the inertial unit vectors in the inertial frame, which are assumed to be unknown. The agents exchange their measurements with neighboring agents according to an undirected and acyclic communication graph $\mathcal{G}$.

Based on the above model and sensing assumptions, we first consider the following leaderless attitude synchronization problems.

\begin{pb}[Kinematics]\label{pb1}
    Consider a network of $N$ agents evolving according to the rotational kinematics \eqref{R_dynamics_i}. Design a distributed feedback control law $\omega_i$ such that, for almost any initial conditions, the orientations of all agents synchronize to a common constant orientation.
\end{pb}

\begin{pb}[Dynamics]\label{pb2}
    Consider a network of $N$ agents evolving according to the rotational dynamics \eqref{R_dynamics_i}-\eqref{w_dynamics_i}. Design a distributed feedback control torque $\tau_i$ such that, for almost any initial conditions, the orientations of all agents synchronize to a common orientation.
\end{pb}

In addition to the leaderless case, we address the problem of attitude synchronization to a constant orientation specified by the designer using reference vectors. In this scenario, we assume that, in addition to the measurements in \eqref{vector_measurement}, reference unit-vector measurements are available only to a single agent, referred to as the leader. Without loss of generality, we index this agent as agent~$1$. The reference unit-vector measurements available to the leader are expressed as
\begin{equation}\label{ref_meas}
b_\ell^r = R_r^\top a_\ell,
\end{equation}
where $R_r \in SO(3)$ represents the desired constant orientation.
The corresponding leader–follower synchronization problems are formulated as follows.

\begin{pb}[Kinematics with Reference]\label{pb3}
    Consider a network of $N$ agents evolving according to the rotational kinematics \eqref{R_dynamics_i}, where only agent~$1$ has access to the reference vector measurements. Design a distributed feedback control law $\omega_i$ such that, for almost any initial conditions, the orientations of all agents converge to the desired constant orientation $R_r$.
\end{pb}

\begin{pb}[Dynamics with Reference]\label{pb4}
    Consider a network of $N$ agents evolving according to the rotational dynamics \eqref{R_dynamics_i}--\eqref{w_dynamics_i}, where only agent~$1$ has access to the reference vector measurements. Design a distributed feedback control torque $\tau_i$ such that, for almost any initial conditions, the orientations of all agents converge to the desired orientation $R_r$.
\end{pb}

The problems formulated above arise naturally in distributed estimation and coordination of orientation in networked sensing systems. In particular, the use of inertial vector measurements and inter-agent communication enables synchronization without requiring direct access to absolute orientation measurements. The leaderless problems in Problems \ref{pb1}-\ref{pb2} focus on achieving consensus on orientation using only relative information, which is essential in scenarios where no global reference frame is available. The leader-follower problems in Problems \ref{pb3}-\ref{pb4} extend this setting to cases where a desired orientation is specified through reference measurements available to a single agent. It is worth noting that restricting the reference information to one agent preserves the distributed nature of the control design and avoids reliance on centralized coordination. Moreover, the assumption of an undirected and acyclic communication graph allows information to propagate through the network while maintaining minimal connectivity requirements.

The problems addressed in this paper are motivated by several practical sensing scenarios encountered in networked systems. In the following subsection, we present representative cases where these problems arise.

\subsection{Body-Frame Position Measurements of Unknown Landmarks}
One representative example is when each agent can measure the body-frame positions of some static landmarks (unknown in the inertial frame). Such measurements naturally occur in vision-based or range-based sensing systems that track environmental features (\eg, see Figure 1). To illustrate, let $\mathrm{l}_\ell \in \mathbb{R}^3$, $\ell = 1,2,\ldots,\bar n$, denote the positions of $\bar n \geq 3$ stationary landmarks expressed in the inertial frame, and let $p_i \in \mathbb{R}^3$ denote the position of agent~$i$ in the inertial frame. Assume that agent~$i$ can measure the position $z_\ell^i$ of each landmark with respect to its body-attached frame. This position can be expressed as follows:
\begin{equation}\label{relative_position_measurement}
    z_\ell^i = R_i^\top (\mathrm{l}_\ell - p_i),
\end{equation}
To remove the dependence on the agent position, consider the relative displacement between two landmarks $\ell$ and $m$, which can be computed locally by agent~$i$ as
\begin{equation}\label{landmark_difference}
    z_{\ell m}^i := z_\ell^i - z_m^i = R_i^\top (\mathrm{l}_\ell - \mathrm{l}_m),\quad \ell \neq m.
\end{equation}
Defining the normalized inertial displacement vectors
\begin{equation}\label{inertial_vectors}
    a_{\ell m} := \frac{\mathrm{l}_\ell - \mathrm{l}_m}{\|\mathrm{l}_\ell - \mathrm{l}_m\|} \in \mathbb{S}^2,
\end{equation}
agent~$i$ can construct the corresponding body-frame unit-vector
\begin{equation}\label{vector_measurement_relative}
    b_{\ell m}^i := \frac{z_{\ell m}^i}{\|z_{\ell m}^i\|} = R_i^\top a_{\ell m}.
\end{equation}
The vectors $a_{\ell m}$ are constant, inertial, and common to all agents, while the corresponding measurements $b_{\ell m}^i$ depend only on the agent’s orientation. When at least two such vectors are non-collinear, the resulting measurements coincide with ~\eqref{vector_measurement} and provide sufficient information for attitude synchronization. Therefore, body-frame position measurements of set of landmarks naturally induce the inertial vector measurements employed in the leaderless and leader--follower attitude synchronization problems formulated in Section~\ref{s3}, without requiring agents to know the inertial vectors and their absolute/relative orientations.

\begin{figure}[h]
    \centering
\includegraphics[width=0.99\linewidth]{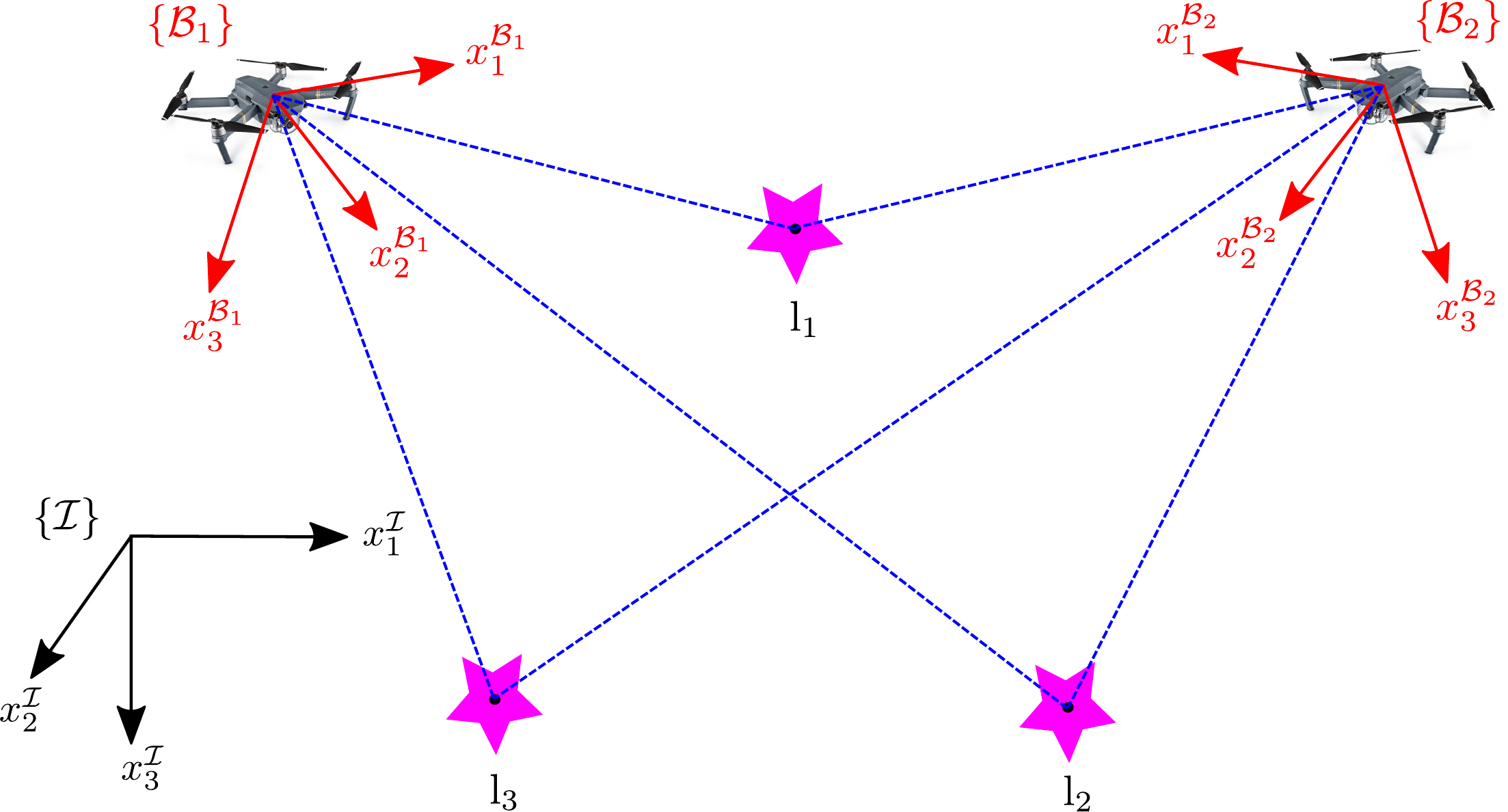}
    \caption{Example scenario where two drones observe the same set of three unknown landmarks.}
    \label{drone_measurements_diagram}
\end{figure}

\subsection{Inertial Vector Measurements from Accelerometers and Magnetometers}
Another practically important sensing scenario arises when each agent is equipped with an inertial measurement unit (IMU) consisting of accelerometers and magnetometers. Such sensors are standard on many robotic platforms and provide direct body-frame measurements of inertial vectors that can be exploited for attitude synchronization.

An accelerometer measures the apparent acceleration expressed in the body-attached frame, which is given by
\begin{equation}\label{specific_force}
    \mathrm{a}^{\mathcal{B}}_i = R_i^\top ( \ddot{p}_i - g ),
\end{equation}
where $\ddot{p}_i \in \mathbb{R}^3$ is the linear acceleration of agent~$i$ expressed in the inertial frame, and $g \in \mathbb{R}^3$ denotes the gravity vector expressed in the inertial frame. Under the common assumption of small accelerations, \ie, $\ddot{p}_i \approx 0$, the accelerometer measurement simplifies to:
\begin{equation}\label{gravity_measurement}
    \mathrm{a}^{\mathcal{B}}_i \approx - R_i^\top g.
\end{equation}
Normalizing this measurement yields the unit vector
\begin{equation}\label{gravity_unit_measurement}
    b_1^i := \frac{\mathrm{a}^{\mathcal{B}}_i}{\|\mathrm{a}^{\mathcal{B}}_i\|} = R_i^\top a_1,
\end{equation}
where $a_1 := -\frac{g}{\|g\|} \in \mathbb{S}^2$ is a constant inertial unit vector common to all agents.

Similarly, a magnetometer measures the Earth’s magnetic field expressed in the body-attached frame. Let $m \in \mathbb{R}^3$ denote the (approximately constant)  Earths magnetic field expressed in the inertial frame. The corresponding magnetometer measurement available to agent~$i$ is given by
\begin{equation}\label{magnetometer_measurement}
    \mathrm{m}^{\mathcal{B}}_i = R_i^\top m.
\end{equation}
By normalizing the magnetic field measurement, we obtain
\begin{equation}
    b_2^i := \frac{\mathrm{m}^{\mathcal{B}}_i}{\|\mathrm{m}^{\mathcal{B}}_i\|} = R_i^\top a_2,
\end{equation}
where $a_2 := \frac{m}{\|m\|} \in \mathbb{S}^2$ is a second inertial unit vector, which is assumed to be non-collinear with $a_1$.
Therefore, under mild motion assumptions and standard sensor normalization, accelerometer and magnetometer measurements naturally provide inertial vector measurements of the form
\begin{equation}
    b_\ell^i = R_i^\top a_\ell, \qquad \ell = 1,2,
\end{equation}
which coincide with the sensing model~\eqref{vector_measurement} adopted in this paper. As a result, the leaderless and leader--follower attitude synchronization problems formulated in Section~\ref{s3} directly apply to networks of agents equipped solely with onboard inertial sensors equipped with an IMU consisting of accelerometers and magnetometers, without requiring explicit attitude reconstruction or absolute orientation measurements.

\section{Leaderless Distributed Attitude Synchronization}\label{leaderless_section}
This section addresses the leaderless attitude synchronization problems formulated in Problems~\ref{pb1} and~\ref{pb2}. The objective is to design fully distributed control laws that drive the orientations of all agents to a common, but a priori unknown, orientation using only local vector measurements and inter-agent communication. No agent has access to a global reference orientation, and coordination is achieved solely through relative information exchanged over the interaction graph. We first consider the kinematic-level design, where the angular velocity is treated as the control input. The proposed feedback law operates directly on the rotation manifold and relies exclusively on vector measurements. We then extend this design to the dynamic level by incorporating the full rigid-body rotational dynamics and proposing distributed control torques that ensure synchronization under both zero and nonzero common angular velocity scenarios.

\subsection{Synchronization at the kinematic level}\label{leaderless_kinematc_design}
This section addresses \textit{Problem \ref{pb1}}, where the design of the feedback control law is performed at the kinematic level of \(SO(3)\), \ie,  treating $ \omega_i$ as the control input. For each $i \in \mathcal{V}$, we consider the following distributed feedback control law:
{
\begin{equation}\label{fc_vector_meas}
    \omega_i = \frac{k_R}{2} \sum_{j \in \mathcal{N}_i} \sum_{l=1}^n \rho_\ell \left(b_\ell^j \times b_\ell^i\right),
\end{equation}}where $k_R, \rho_\ell >0$. A similar distributed leaderless attitude synchronization scheme has been proposed in \cite{Sarlette_SIAM2012}, but it is endowed only with local stability guarantees. We assume that the gains $\rho_\ell$ are chosen such that the matrix $A:=\sum_{l=1}^n \rho_\ell a_\ell a_\ell^\top$ has three distinct eigenvalues. Note that under the assumption that at least two inertial vectors are non-collinear, it can be verified that the matrix $A$ is semi-positive definite. To address   \textit{Problem \ref{pb1}}, one should demonstrate that the feedback control law \eqref{fc_vector_meas} ensures, for every \(i \in \mathcal{V}\) and \(j \in \mathcal{N}_i\), that the relative orientation \(R_j^\top R_i\) converges to the identity matrix. Each edge in the graph has two possible relative orientations. Specifically, for every \((i, j) \in \mathcal{E}\), both \(R_j^\top R_i\) and \(R_i^\top R_j\) are defined. However, the convergence of one orientation inherently ensures the convergence of the other. To simplify the stability analysis and avoid redundancy, we consider only one relative orientation for each edge. This is accomplished by assigning an arbitrary virtual orientation to the graph \(\mathcal{G}\) and indexing each oriented edge with an integer. Consequently, for any two agents \(i\) and \(j\) connected by an oriented edge \(k\), the relative attitude is defined as $\bar{R}_k := R_j R_i^\top$ where \(\{k\} = \mathcal{M}_i^+ \cap \mathcal{M}_j^- \subset \mathcal{M}\). From \eqref{R_dynamics_i}, one can derive the following dynamics for $\bar R_k$:

{\small
\begin{equation}
     \dot{\bar R}_k = \bar R_k[\bar \omega_k]^{\times}, \label{R_bar_dynamics_k}
\end{equation}}where \(\{k\} = \mathcal{M}_i^+ \cap \mathcal{M}_j^-\) and $\bar \omega_k :=  R_i (\omega_j - \omega_i)$ for every $(i, j) \in \mathcal{E}$. Note that for each \( i \in \mathcal{V} \) and \( j \in \mathcal{N}_i \), the intersection \( \mathcal{M}_i^+ \cap \mathcal{M}_j^- \) is non-empty if and only if there exists an oriented edge \(k\) from \( i \) to \( j \). In such a case, the intersection contains exactly that edge (\ie, \( \mathcal{M}_i^+ \cap \mathcal{M}_j^- = \{k\} \)). If no such edge exists, the intersection is empty. Define $\bar \omega =[\bar \omega_1^\top , \bar \omega_2^\top , \hdots, \bar \omega_M^\top ]^\top  \in \mathbb{R}^{3M}$ and $\omega =[\omega_1^\top , \omega_2^\top , \hdots, \omega_N^\top ]^\top  \in \mathbb{R}^{3N}$. One can derive the following equation that relates $\bar \omega$ and $\omega$:

{\small
\begin{equation}\label{w_bar}
    \bar \omega = -\mathbf{H}^\top\mathbf{R}\omega,
\end{equation}}where $\mathbf{R}:=\text{diag}(R_1, R_2, \hdots, R_N) \in \mathbb{R}^{3N \times 3N}$ and $\mathbf{H}$ is defined as follows:
{\small
\begin{equation}\label{H_bar}
   \mathbf{H}(t):=[H_{ik}]_{N\times M} \hspace{0.3cm} \text{with} \hspace{0.3cm} H_{ik}=\begin{cases}
      I_3 & k\in\mathcal{M}_i^+\\
      -\bar R_k & k\in\mathcal{M}_i^-\\
      0 & \text{otherwise}
    \end{cases}.   
\end{equation}}It is important to note that the assigned orientation of the graph $\mathcal{G}$ is purely notional and does not alter the undirected nature of the interaction graph $\mathcal{G}$. Considering a single relative orientation between each pair of neighboring agents, defined by the virtual orientation assigned to the graph, attitude synchronization is achieved when $\bar{R}_k = I_3$ for each $k \in \mathcal{M}$. Next, we analyze the stability properties of the dynamics \eqref{R_dynamics_i}, considering the feedback control signal given by \eqref{fc_vector_meas}. Before proceeding, we define the set $\mathcal{A} := \left\{ x \in \mathcal{S} : \forall k \in \mathcal{M}, \; \bar{R}_k = I_3\right\}$, where $x := \left( \bar R_1, \bar R_2, \ldots, \bar R_M\right) \in \mathcal{S}$ and $\mathcal{S}:= SO(3)^M$. We are now ready to state the main result of this section.

\begin{thm}\label{theorem_1}
    Let a network of $N$ agents rotate according to the kinematics given in \eqref{R_dynamics_i}. Assume that the measurement \eqref{vector_measurement} is available and the interaction graph $\mathcal G$ is an undirected tree. Consider the dynamics \eqref{R_bar_dynamics_k} with feedback control \eqref{fc_vector_meas}. Then, the following statements hold:
    \begin{enumerate}[i)]
    \item All solutions of \eqref{R_bar_dynamics_k} with \eqref{fc_vector_meas} converge to the set of equilibria $\Upsilon : =\mathcal{A} \cup \{x \in \mathcal{S} : \bar{R}_m = I_3, \, \bar R_n=\mathcal{R}(\pi, u_{\beta_n}), \, \forall m \in \mathcal{M}^I, \, \forall n \in \mathcal{M}^\pi\}$, where $\mathcal{M}^I \cup \mathcal{M}^\pi=\mathcal{M}$, $|\mathcal{M}^\pi|>0$, $\beta_n \in \{1, 2, 3\}$, and $u_{\beta_n} \in \mathcal{E}(A)$ with $\mathcal{E}(A)\subset \mathbb{S}^2$ denotes the set of unit eigenvectors of matrix $A$.\label{set_of_equilibrium}
    \item The set of all undesired equilibrium points $\Upsilon \setminus \mathcal{A}$ is unstable.\label{unstability_of_equilibrium}
    \item The desired equilibrium set $\mathcal{A}$ is \textit{almost globally asymptotically stable}\footnote{ The set $\mathcal{A}$ is said to be almost globally asymptotically stable if it is asymptotically stable, and attaractive from all initial conditions except a set of zero Lebesgue measure.}. \label{stability_of_equilibrium}
    \end{enumerate}
\end{thm}
\begin{proof}
    See Appendix \ref{app_1}
\end{proof}

Theorem~\ref{theorem_1} shows that the set of desired equilibria $\mathcal{A}$ is almost globally asymptotically stable for system~\eqref{R_bar_dynamics_k} under the control law~\eqref{fc_vector_meas}. Unfortunately, as discussed in~\cite{Koditschek1_989}, no smooth feedback control law can achieve global asymptotic stabilization on $SO(3)$ due to topological obstructions. Consequently, the presence of unstable equilibrium points in the set $\Upsilon \setminus \mathcal{A}$ is unavoidable when using smooth vector fields.

\subsection{Synchronization at the dynamic level}
To address \textit{Problem~\ref{pb2}}, we extend the kinematic-level synchronization strategy presented in the previous section to the full rigid-body rotational dynamics on \(SO(3)\). At this level, the angular velocity is no longer treated as a virtual control input, and the control design must explicitly account for inertial effects and nonlinear Coriolis terms.
Depending on the desired steady-state behavior, two synchronization objectives are considered. The first aims at achieving attitude synchronization to a constant orientation by driving all angular velocities to zero. The second allows convergence to a common nonzero angular velocity, resulting in synchronization to a time-varying orientation. To accommodate both cases, we propose the following distributed feedback control torque for each agent \(i \in \mathcal{V}\):
\begin{align}\label{torque_i}
    \tau_i &= [\omega_i]^\times J_i \omega_i+\frac{k_R}{2} \sum_{j \in \mathcal{N}_i} \sum_{l=1}^n \rho_\ell \left(b_\ell^j \times b_\ell^i\right)-k_\omega \omega_i\nonumber\\
    &~~~~~~~~~~~~~~~~~~~~~~-\bar k_\omega \sum_{j \in \mathcal{N}_i}(\omega_i-\omega_j),
\end{align}
where $k_\omega$ and $\bar{k}_\omega$ are non-negative scalars. 
The first term in the control torque \eqref{torque_i} compensates for the nonlinear Coriolis effects in the dynamics of $\omega_i$, while the remaining three terms drive the agents' orientations and angular velocities toward common values. Figure \ref{Controller_structure_less} illustrates the
structure of the proposed distributed feedback control law \eqref{torque_i}. As we will discuss later, choosing $k_\omega > 0$ and $\bar{k}_\omega \geq 0$ introduces the necessary damping to drive the agents' angular velocities to zero, leading to synchronization at a constant orientation. Setting $k_\omega = 0$ and $\bar{k}_\omega > 0$ results in damping that ensures convergence to a nonzero common angular velocity, leading to synchronization at a time-varying orientation.

Define the extended state $\bar x := (x,\omega) \in \bar{\mathcal{S}}$, where 
$\bar{\mathcal{S}} := SO(3)^M \times \mathbb{R}^{3N}$. In the following theorem we will establish the stability properties of the dynamics \eqref{R_bar_dynamics_k} and \eqref{w_dynamics_i} under the distributed feedback torque \eqref{torque_i}, for the case $k_\omega>0$ and $\bar{k}_\omega \geq 0$.  

\begin{thm}\label{theorem_2}
    Let a network of $N$ agents rotate according to the dynamics given in \eqref{R_dynamics_i}-\eqref{w_dynamics_i}. Assume that the measurement \eqref{vector_measurement} is available and the interaction graph $\mathcal G$ is an undirected tree. Consider the dynamics \eqref{w_dynamics_i} and \eqref{R_bar_dynamics_k} under the control torque \eqref{torque_i} with $k_R>0$, $k_\omega>0$, and $\bar k_\omega \geq 0$. Then, the following statements hold:
    \begin{enumerate}[i)]
    \item All solutions of \eqref{w_dynamics_i} and \eqref{R_bar_dynamics_k} with \eqref{torque_i} converge to the set of equilibria $\bar{\Upsilon}_0 :=\bar{\mathcal{A}}_0 \cup \{x \in \mathcal{S} : \omega=0, \, \bar{R}_m = I_3, \, \bar R_n=\mathcal{R}(\pi, u_{\beta_n}), \, \forall m \in \mathcal{M}^I, \, \forall n \in \mathcal{M}^\pi\}$, where $\bar{\mathcal{A}}_0:=\{\bar x \in \bar{\mathcal{S}}: x \in \mathcal{A}, ~\omega=0\}$. \label{dyn_set_of_equilibrium}
    \item The set of all undesired equilibrium points $\bar \Upsilon_0 \setminus \bar{\mathcal{A}}_0$ is unstable.\label{dyn_unstability_of_equilibrium}
    \item The desired equilibrium set $\bar{\mathcal{A}}_0$ is \textit{almost globally asymptotically stable}. \label{dyn_stability_of_equilibrium}
    \end{enumerate}
\end{thm}
\begin{proof}
    See Appendix \ref{app_2}
\end{proof}

Theorem~\ref{theorem_2} establishes that the desired equilibrium set $\bar{\mathcal{A}}_0$ is almost globally asymptotically stable for the closed-loop dynamics \eqref{w_dynamics_i} and \eqref{R_bar_dynamics_k} under the distributed control torque~\eqref{torque_i}. As in the kinematic-level case, this result is fundamentally constrained by the topology of the rotation manifold $SO(3)$, and global asymptotic stabilization of rotations cannot be achieved using smooth feedback laws, even when the design is performed at the dynamic level \cite{Koditschek1_989}. Consequently, the existence of unstable equilibrium configurations, represented by the set $\bar{\Upsilon}_0 \setminus \bar{\mathcal{A}}_0$, is unavoidable.
Unlike Theorem~\ref{theorem_2}, the next result addresses the case where \(k_\omega = 0\) and \(\bar{k}_\omega > 0\), for which the closed-loop system synchronizes to a common angular velocity that is not necessarily zero.

\begin{pro}\label{theorem_3}
    Let a network of $N$ agents rotate according to the dynamics given in \eqref{R_dynamics_i}-\eqref{w_dynamics_i}. Assume that the measurement \eqref{vector_measurement} is available and the interaction graph $\mathcal G$ is an undirected tree. Consider the dynamics \eqref{w_dynamics_i} and \eqref{R_bar_dynamics_k} under the control torque \eqref{torque_i} with $k_R>0$, $k_\omega=0$, and $\bar k_\omega > 0$. Then, the following statements hold:
    \begin{enumerate}[i)]
    \item All solutions of \eqref{w_dynamics_i} and \eqref{R_bar_dynamics_k} with \eqref{torque_i} are bounded and converge to the set of equilibria $\bar{\Upsilon}_c :=\bar{\mathcal{A}}_c \cup \{x \in \mathcal{S} : \omega=\textbf{1}_N \otimes \Omega, \, \bar{R}_m = I_3, \, \bar R_n=\mathcal{R}(\pi, u_{\beta_n}), \, \forall m \in \mathcal{M}^I, \, \forall n \in \mathcal{M}^\pi\}$, where $\bar{\mathcal{A}}_c:=\{\bar x \in \bar{\mathcal{S}}: x \in \mathcal{A}, ~\omega=\textbf{1}_N \otimes \Omega\}$ and $\Omega$ is a vector in $\mathbb{R}^3$. \label{dyn_c_set_of_equilibrium}
    \item There exists a positive constant $c > 0$ such that, for any initial condition satisfies
    \begin{equation}
        k_R \sum_{k=1}^{M} \text{tr}\left(A\left(I_3-\bar{R}_k(0)\right)\right)+\sum_{i=1}^{N} \omega_i(0)^\top J_i\, \omega_i(0) < c,\label{initial_cond}
    \end{equation}
    the trajectories of \eqref{w_dynamics_i} and \eqref{R_bar_dynamics_k} under the control law \eqref{torque_i} are bounded and converge to the desired set $\bar{\mathcal{A}}_c$.\label{converge_item}
    \end{enumerate}
\end{pro}
\begin{proof}
    See Appendix \ref{app_3}
\end{proof}

In contrast to Theorem \ref{theorem_2}, which guarantees almost global asymptotic stability for synchronization to a constant orientation, Proposition \ref{theorem_3} establishes only local convergence to a synchronized attitude motion characterized by a common nonzero angular velocity. This limitation is the price of achieving synchronization along a time-varying orientation with the proposed design. 

\begin{rmk}
The distinction between Theorem~\ref{theorem_2} and Proposition~\ref{theorem_3}, despite their differing stability properties, illustrates the flexibility and modularity of the proposed distributed control framework. By appropriately selecting the damping gains $k_\omega$ and $\bar{k}_\omega$, the same control structure can accommodate different synchronization objectives. Specifically, when $k_\omega > 0$, the injected damping drives all angular velocities to zero, achieving synchronization to a constant orientation. Conversely, when $k_\omega = 0$ and $\bar{k}_\omega > 0$, the control law enforces agreement among the agents’ angular velocities without eliminating their common motion, resulting in synchronization along a common time-varying trajectory. Notably, both behaviors are realized in a fully distributed manner, relying solely on local measurements and inter-agent communication, without requiring a leader or a global reference frame. This gain-dependent transition between constant and time-varying attitude synchronization highlights the versatility of the proposed design and allows the control objective to be adjusted to the application without modifying the underlying control architecture.
\end{rmk}

\section{Leader-Follower Distributed Attitude Synchronization}\label{leader-follower_section}
In this section, we extend the leaderless synchronization framework to address the leader–follower attitude synchronization problems formulated in Problems~\ref{pb3} and~\ref{pb4}. Unlike the previous section, the objective here is to synchronize all agents to a prescribed constant orientation specified through reference vector measurements available to a single agent, referred to as the leader. The remaining agents do not have direct access to the reference information and must rely on local communication to achieve alignment. We propose distributed control laws that preserve the decentralized structure of the network while allowing the reference information to propagate through the interaction graph. As in the leaderless case, we first present a kinematic-level design and establish its almost global stability properties. The results are then extended to the dynamic level by designing distributed control torque that guarantees synchronization under the full rigid-body rotational dynamics.

\subsection{Leader-follower synchronization at the kinematic level}
Following the design philosophy introduced in Section \ref{leaderless_kinematc_design}, we propose a distributed control law for the angular velocity $\omega_i$ of each agent to achieve the objective stated in Problem \ref{pb3}. For each $i \in \mathcal{V}$, the angular velocity dynamics are chosen as
\begin{equation}\label{w_leader} 
\omega_i =
\frac{k_R}{2} \sum_{j \in \mathcal{N}_i} \sum_{\ell=1}^n \rho_\ell \left(b_\ell^j \times b_\ell^i\right)
+\frac{k_i}{2}\sum_{\ell=1}^n \bar{\rho}_\ell \left(b_\ell^r \times b_\ell^i\right),
\end{equation}
where $\bar{\rho}_\ell>0$, $k_1>0$, and $k_i=0$ for all $i \in \mathcal{V}\setminus\{1\}$. The gains $\bar{\rho}_\ell$ are selected such that the matrix $ \bar A := \sum_{l=1}^n \bar{\rho}_\ell a_\ell a_\ell^\top$ has three distinct eigenvalues. The control law \eqref{w_leader} consists of two components. The first term coincides with the leaderless synchronization law and enforces agreement among neighboring agents. The second term is active only for the leader and injects the reference information into the network through reference unit-vector measurements. 
The objective of Problem~\ref{pb3} is achieved if the relative orientations $R_j^\top R_i$ and $R_r^\top R_1$ converge to the identity matrix for all $i \in \mathcal{V}$ and $j \in \mathcal{N}_i$. In analogy with the relative attitude errors defined between neighboring agents, we define the attitude error between the reference orientation and the leader's orientation as $\tilde{R}_1 := R_r R_1^\top$. Using the kinematics~\eqref{R_dynamics_i}, the dynamics of $\tilde{R}_1$ evolve according to
\begin{equation}\label{R1_tilde}
    \dot{\tilde{R}}_1 = -\tilde{R}_1 [R_1 \omega_1]^\times,
\end{equation}
where $i \in \mathcal{V}$. Now, let $\mathcal{A}^r :=
\left\{
x^r \in \mathcal{S}^r \,:\,
\bar{R}_k = I_3,\ \forall k \in \mathcal{M},\;
\tilde{R}_1 = I_3
\right\}$,
where $x^r := (\bar R_1,\bar R_2,\ldots,\bar R_M,\tilde{R}_1)$ and $\mathcal{S}^r := SO(3)^{M+1}$. The following theorem characterizes the stability properties of the closed-loop system \eqref{R_bar_dynamics_k} and \eqref{R1_tilde} with feedback control \eqref{w_leader}.

\begin{thm}\label{theorem_4}
    Let a network of $N$ agents evolve according to the kinematics \eqref{R_dynamics_i}. Suppose that the vector measurements \eqref{vector_measurement} are available to all agents, the reference unit-vector measurements \eqref{ref_meas} are available to agent $1$, and the interaction graph $\mathcal G$ is an undirected tree. Consider the closed-loop dynamics \eqref{R_bar_dynamics_k} and \eqref{R1_tilde} under the feedback control law \eqref{w_leader}. Then, the following statements hold:
    \begin{enumerate}[i)]
    \item All solutions of \eqref{R_bar_dynamics_k} and \eqref{R1_tilde}, under the control law \eqref{w_leader}, converge to the equilibrium set
    \begin{align}
        \Upsilon^r := \mathcal{A}^r \,\cup\, &\Bigl\{ x^r \in \mathcal{S}^r \,\big|\, 
        \tilde{R}_1 \in \{ I_3, \mathcal{R}(\pi,\bar{u}) \},\;
        \bar{R}_m = I_3,\;\nonumber\\
        &\bar{R}_n = \mathcal{R}(\pi, u_{\beta_n}),\;
        \forall m \in \mathcal{M}^I,\;
        \forall n \in \mathcal{M}^\pi
        \Bigr\},\nonumber
    \end{align}
    where the index sets satisfy $\mathcal{M}^I \cup \mathcal{M}^\pi = \mathcal{M}$, with either $|\mathcal{M}^\pi| > 0$, or $|\mathcal{M}^\pi| = 0$ and $\tilde{R}_1 = \mathcal{R}(\pi, \bar{u})$. Moreover, $ \forall \beta_n \in \{1,2,3\}$, $u_{\beta_n} \in \mathcal{E}(A)$, and $\bar{u} \in \mathcal{E}(\bar{A})$, where $\mathcal{E}(A) \subset \mathbb{S}^2$ and $\mathcal{E}(\bar{A}) \subset \mathbb{S}^2$ denote the sets of unit eigenvectors of the matrices $A$ and $\bar{A}$, respectively. \label{set_of_equilibrium_leader}
    \item The set of all undesired equilibrium points $\Upsilon^r \setminus \mathcal{A}^r$ is unstable.\label{unstability_of_equilibrium_leader}
    \item The desired equilibrium set $\mathcal{A}^r$ is \textit{almost globally asymptotically stable}. \label{stability_of_equilibrium_leader}
    \end{enumerate}
\end{thm}
\begin{proof}
    See Appendix \ref{app_4}
\end{proof}

Theorem~\ref{theorem_4} establishes that the desired equilibrium set $\mathcal{A}^r$ is almost globally asymptotically stable under the distributed leader--follower control law \eqref{w_leader}. As in the leaderless case, global asymptotic stabilization on $SO(3)$ cannot be achieved using smooth feedback due to inherent topological obstructions \cite{Koditschek1_989}. 
Moreover, although the reference unit-vector measurements are available only to the leader, the connectivity of the interaction graph ensures that the reference information propagates to all followers through local relative measurements, resulting in global alignment without centralized coordination or explicit dissemination of a global reference.

\subsection{Leader-follower synchronization at the dynamic level}
In this subsection, we extend the leader-follower attitude synchronization framework to the full rigid-body attitude dynamics.
Building upon the kinematic-level design developed in the previous section, our objective is to achieve attitude synchronization to a prescribed reference orientation in the presence of rotational dynamics, inertial effects, and angular velocity coupling. To this end, we design a distributed feedback control torque that preserves the leader-follower structure, injects the reference information through the leader only, and guarantees almost global asymptotic stability. For each $i \in \mathcal{V}$, we propose the following distributed feedback control torque:
\begin{align}\label{torque_i_leader}
    \tau_i &= \frac{k_R}{2} \sum_{j \in \mathcal{N}_i} \sum_{\ell=1}^n \rho_\ell \left(b_\ell^j \times b_\ell^i\right)
    +\frac{k_i}{2}\sum_{\ell=1}^n \bar{\rho}_\ell \left(b_\ell^r \times b_\ell^i\right)\nonumber\\
    &\quad +[\omega_i]^\times J_i \omega_i
    -k_\omega \omega_i
    -\bar k_\omega \sum_{j \in \mathcal{N}_i}(\omega_i-\omega_j),
\end{align}
where $\bar{\rho}_\ell>0$, $k_i>0$ for $i=1$, and $k_i=0$ for all $i\in\mathcal{V}\setminus\{1\}$. Figure \ref{Controller_structure_1} illustrates the
structure of the proposed distributed feedback control law \eqref{torque_i_leader}. 
Note that, as in the control law \eqref{w_leader}, the second term in \eqref{torque_i_leader} is active only for the leader and injects the reference information into the network through the reference vector measurements, whereas the first and last terms are fully distributed. Define the set $\bar{\mathcal{A}}_0^r :=\left\{\bar x^r \in \bar{\mathcal{S}}^r : \forall k \in \mathcal{M},\; \bar{R}_k = I_3,\; \tilde{R}_1=I_3,\; \omega=0 \right\}$, where $\bar x^r := \left( \bar R_1, \bar R_2, \ldots, \bar R_M, \tilde{R}_1, \omega\right) \in \bar{\mathcal{S}}^r$
and $\bar{\mathcal{S}}^r := SO(3)^{M+1}\times \mathbb{R}^{3N}$. The following theorem establishes the equilibrium characterization and almost global asymptotic stability of the closed-loop system \eqref{w_dynamics_i}, \eqref{R_bar_dynamics_k}, and \eqref{R1_tilde} under the proposed control torque \eqref{torque_i_leader}.
\begin{thm}\label{theorem_5}
    Consider a network of $N$ agents, each evolving according to the rotational dynamics defined in \eqref{R_dynamics_i}-\eqref{w_dynamics_i}. Suppose the vector measurements \eqref{vector_measurement} are available to all agents, while only agent $1$ has access to the reference unit-vector measurements \eqref{ref_meas}. The interaction graph $\mathcal G$ is assumed to be an undirected tree. Under the closed-loop dynamics given by \eqref{w_dynamics_i}, \eqref{R_bar_dynamics_k}, and \eqref{R1_tilde}, and the control torque described in \eqref{torque_i_leader} with gains satisfying $k_R > 0$, $k_\omega > 0$, and $\bar k_\omega \geq 0$, the following statements hold:
    \begin{enumerate}[i)]
    \item All solutions of \eqref{w_dynamics_i}, \eqref{R_bar_dynamics_k} and \eqref{R1_tilde} with \eqref{torque_i_leader} converge to the set of equilibria $\bar{\Upsilon}^r_0 :=\bar{\mathcal{A}}^r_0 \cup \{\bar x^r \in \bar{\mathcal{S}}^r : \omega=0,\, 
        \tilde{R}_1 \in \{ I_3, \mathcal{R}(\pi,\bar{u}) \},\; \, \bar{R}_m = I_3, \, \bar R_n=\mathcal{R}(\pi, u_{\beta_n}), \, \forall m \in \mathcal{M}^I, \, \forall n \in \mathcal{M}^\pi\}$.\label{dyn_set_of_equilibrium_leader}
    \item The set of all undesired equilibrium points $\bar \Upsilon^r_0 \setminus \bar{\mathcal{A}}^r_0$ is unstable.\label{dyn_unstability_of_equilibrium_leader}
    \item The desired equilibrium set $\bar{\mathcal{A}}^r_0$ is \textit{almost globally asymptotically stable}. \label{dyn_stability_of_equilibrium_leader}
    \end{enumerate}
\end{thm}
\begin{proof}
    See Appendix \ref{app_5}
\end{proof}

\begin{rmk}
It is important to note that, due to the structure of the control laws \eqref{w_leader} and \eqref{torque_i_leader}, which include terms that drive the agents' orientations towards consensus, and because the leader's orientation converges to the reference orientation $R_r$, the orientations of all other agents will also converge to $R_r$. This convergence occurs because the interaction graph is connected, ensuring that each agent's orientation is coupled to the leader's through local communication. As a result, the leader's orientation serves as the global reference, driving all agents' attitudes to align with $R_r$ as the leader's attitude converges to the reference orientation.
\end{rmk}

\begin{figure*}[t]
    \centering
    \subfigure[]{\includegraphics[width=0.48\textwidth]{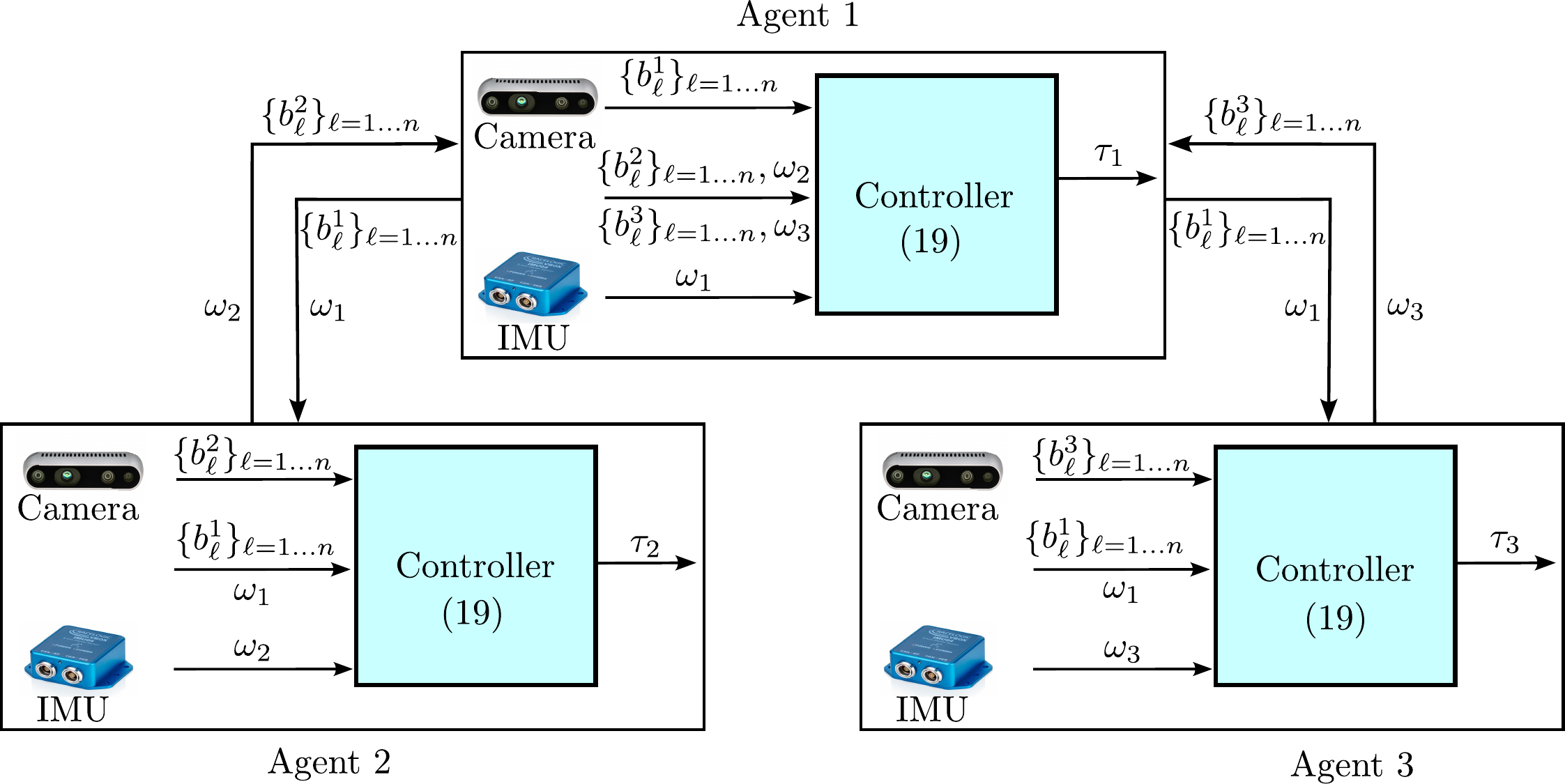}\label{Controller_structure_less}} 
    \hfill
    \subfigure[]{\includegraphics[width=0.48\textwidth]{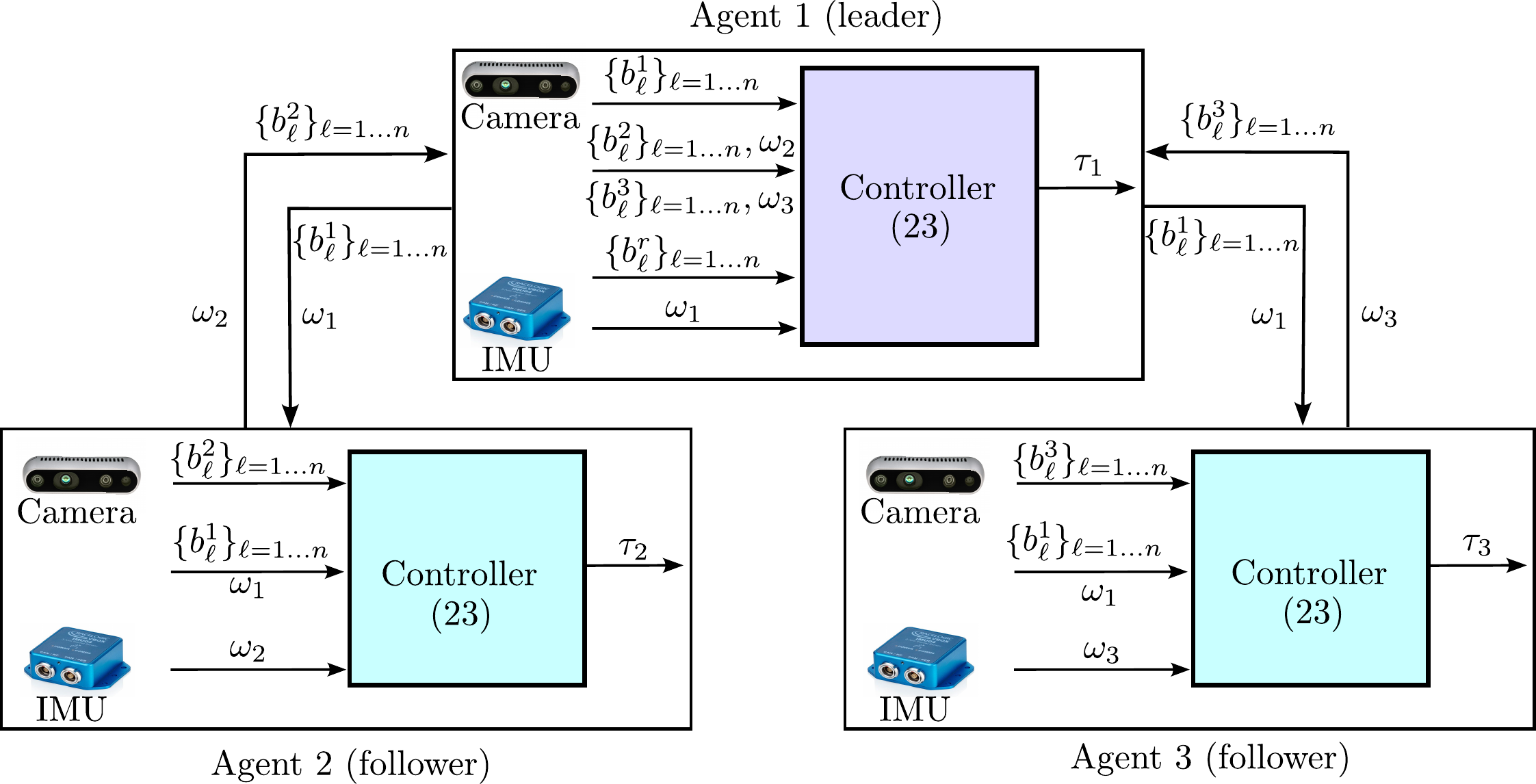}\label{Controller_structure_1}}
    \caption{Proposed control schemes at the dynamic level for a three-agent network: (a) Leaderless scheme \eqref{torque_i} (b) Leader-follower scheme \eqref{torque_i_leader}.}
    \label{Controller_structure}
\end{figure*}

\section{SIMULATION}\label{s6}
This section evaluates the performance of the proposed distributed synchronization schemes \eqref{fc_vector_meas} and \eqref{torque_i} through numerical simulations. Along with the leaderless synchronization approach, the leader–follower schemes \eqref{w_leader} and \eqref{torque_i_leader} are also examined\footnote{Simulation videos can be found at  
\href{https://youtu.be/mpx20gwkoY4}{\textcolor{Rhodamine}{https://youtu.be/mpx20gwkoY4}}.}. Consider a network of eight rigid-boy systems, \ie, $N=8$, interacting according to the undirected graph topology depicted in Fig.~\ref{graph}. 
\begin{figure}[H]
    \centering
    \includegraphics[width=0.99\linewidth]{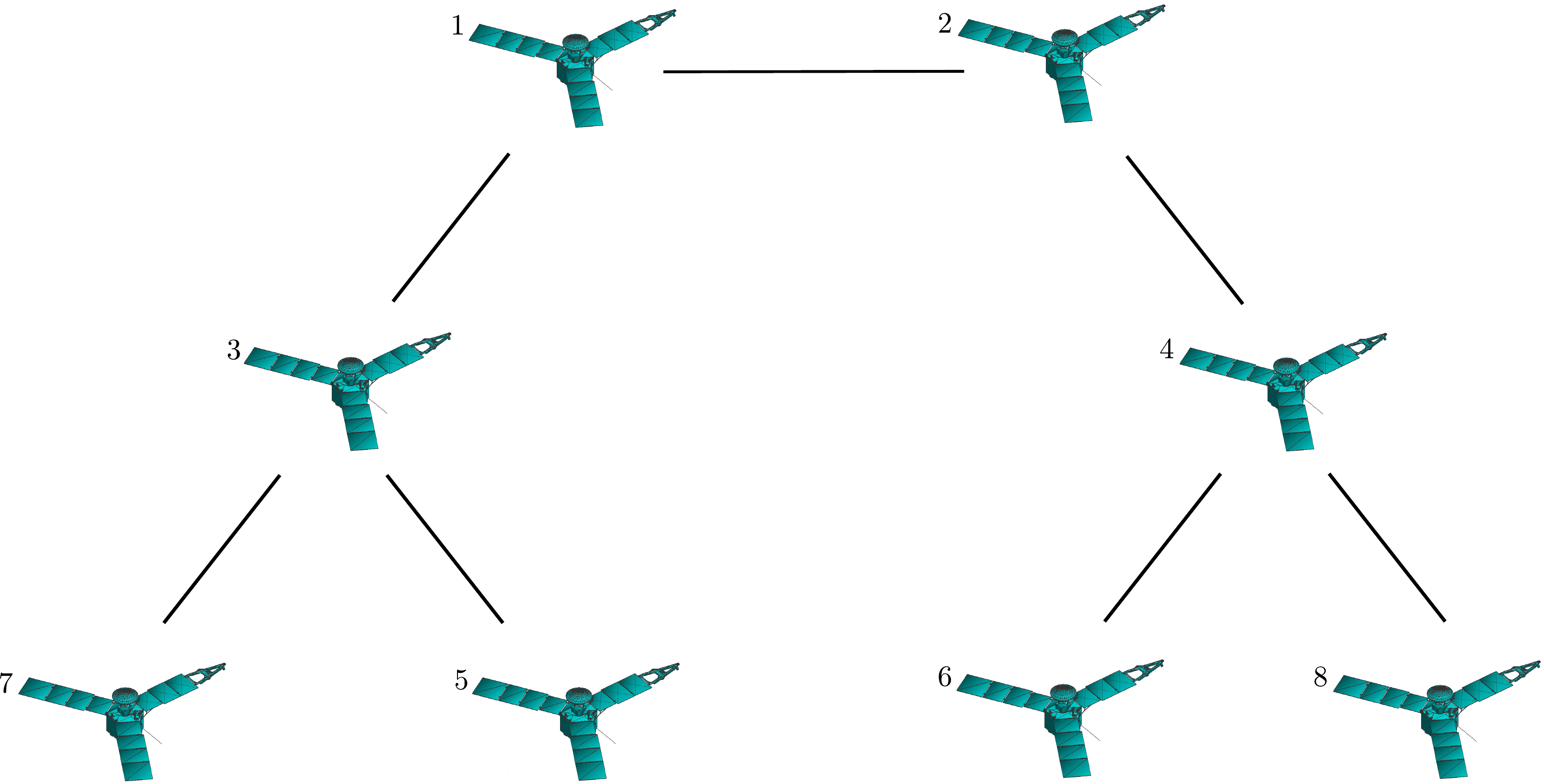}
    \caption{Interaction graph for a network of eight rigid-boy systems.}
    \label{graph}
\end{figure}
Without loss of generality, all satellites are assumed to have identical inertia matrices given by $J = \text{diag}(0.0159,\, 0.015,\, 0.0297)$. For Leaderless synchronization, three simulations are conducted. The first evaluates the kinematic control law \eqref{fc_vector_meas}. The second and third simulations assess the dynamic controller \eqref{torque_i} under two gain configurations: 
$(i)\; k_\omega > 0,\; \bar{k}_\omega \geq 0,$ and $(ii)\; k_\omega = 0,\; \bar{k}_\omega > 0$. For all simulations, the initial attitudes are chosen as rotations about the axis $u = [1~0~0]^\top$, specified as follows: $R_1(0) = \mathcal{R}(\tfrac{\pi}{10},u),~
    R_2(0) = \mathcal{R}(\tfrac{9\pi}{10},u),~
    R_3(0) = \mathcal{R}(\tfrac{4\pi}{10},u),~
    R_4(0) = \mathcal{R}(\tfrac{3\pi}{10},u),~
    R_5(0) = \mathcal{R}(\tfrac{2\pi}{10},u),~
    R_6(0) = \mathcal{R}(\tfrac{8\pi}{10},u),~
    R_7(0) = \mathcal{R}(\tfrac{7\pi}{10},u),~
    R_8(0) = \mathcal{R}(\tfrac{6\pi}{10},u)$.
The initial angular velocities are randomly selected as follows: $\omega_1(0)=[0.1~0.6~0.6]^\top,~
    \omega_2(0)=[0.4~0.95~0.87]^\top,~
    \omega_3(0)=[0.73~0.69~0.58]^\top,~
    \omega_4(0)=[0~0.87~0]^\top,~
    \omega_5(0)=[0.45~0.18~0.48]^\top,~
    \omega_6(0)=[0.74~0~1]^\top,~
    \omega_7(0)=[0.5~0.7~0.94]^\top,~
    \omega_8(0)=[0.69~0.73~0.5]^\top$.
Two inertial reference vectors are used: $a_1 = [1~0~0]^\top$ and $a_2 = [0~0~1]^\top$, with $n=2$, $\rho_1=1$, $\rho_2=2$, and $k_R = 1$. The corresponding results are shown in Figs.~\ref{sim}, \ref{sim_w0}, and \ref{sim_wc}.  
Fig.~\ref{sim} illustrates that control law \eqref{fc_vector_meas} achieves attitude synchronization to a constant common orientation. Figs.~\ref{sim_w0} and \ref{sim_wc} show the performance of controller \eqref{torque_i}. When $k_\omega>0$ and $\bar{k}_\omega=1$, both the attitudes and angular velocities converge to common constant values. When $k_\omega=0$ and $\bar{k}_\omega=1$, the satellites synchronize to a common nonzero angular velocity, resulting in convergence to a common time-varying attitude, as predicted by the theoretical analysis.

\begin{figure}[!]
    \centering
    \includegraphics[width=0.99\linewidth]{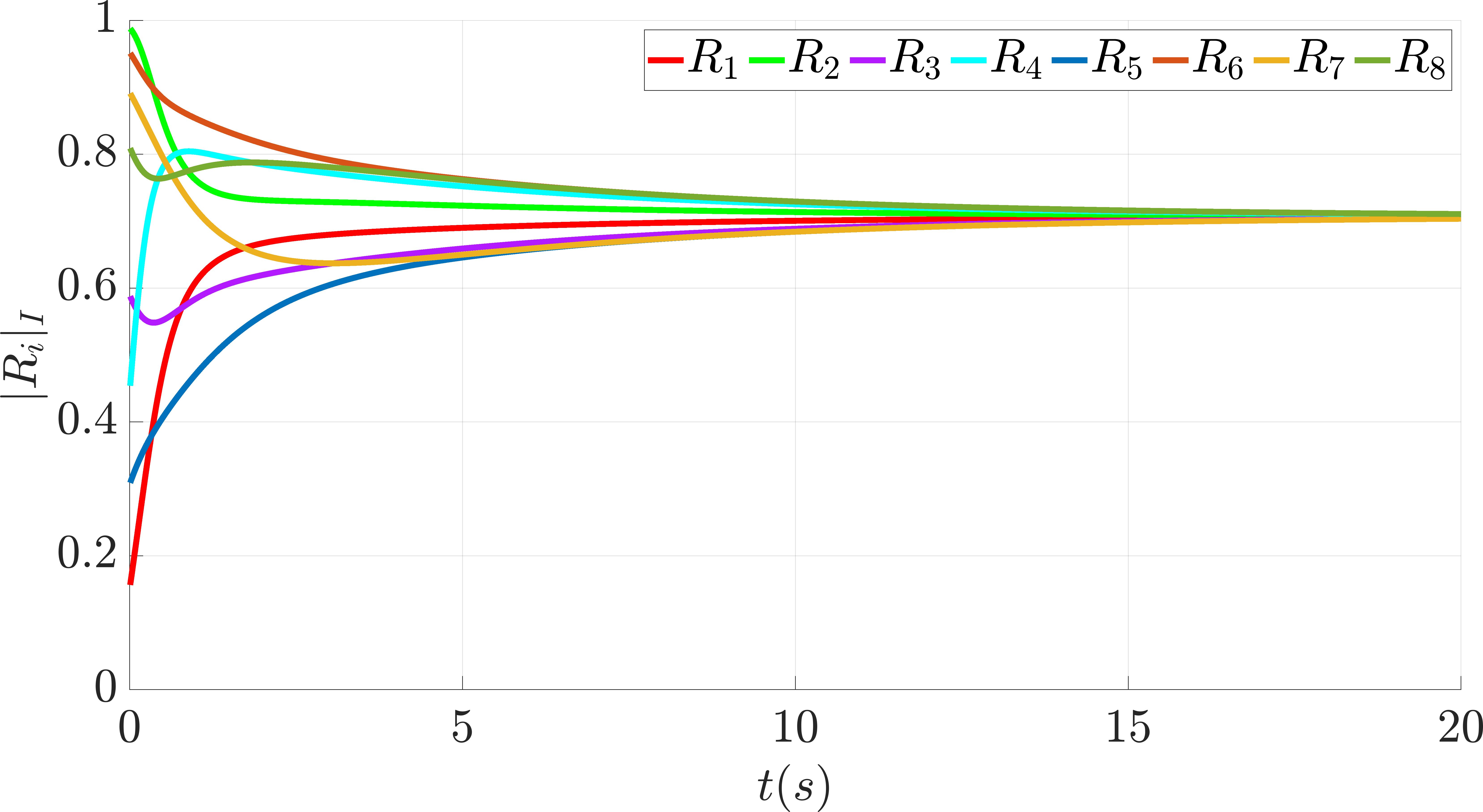}
    \caption{Time evolution of the attitudes under the control \eqref{fc_vector_meas}.}
    \label{sim}
\end{figure}

\begin{figure}[!]
    \centering
    \includegraphics[width=0.99\linewidth]{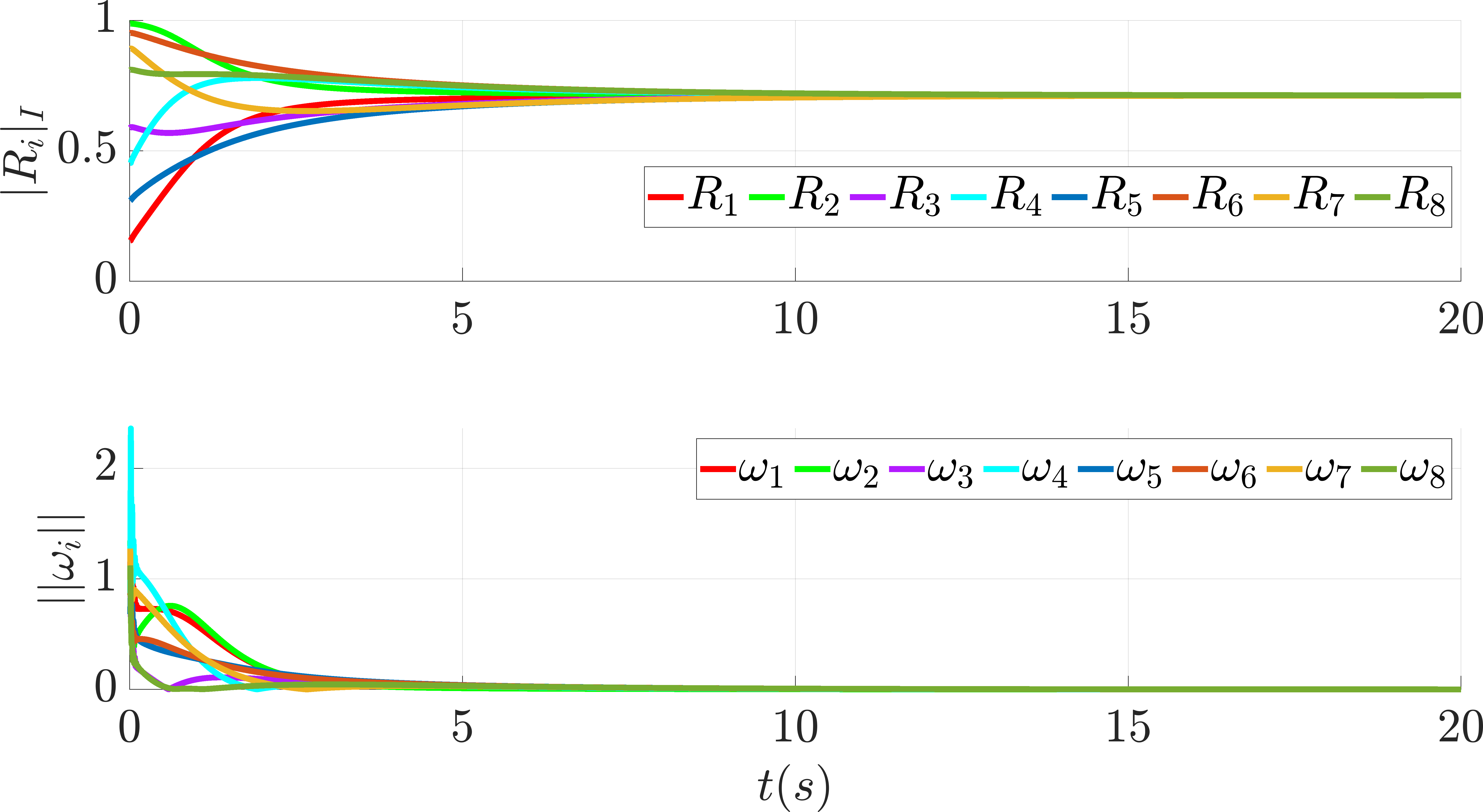}
    \caption{Time evolution of the attitudes and angular velocities under the control \eqref{torque_i}, with $k_\omega=1$ and $\bar{k}_\omega=1$.}
    \label{sim_w0}
\end{figure}

\begin{figure}[!]
    \centering
    \includegraphics[width=0.99\linewidth]{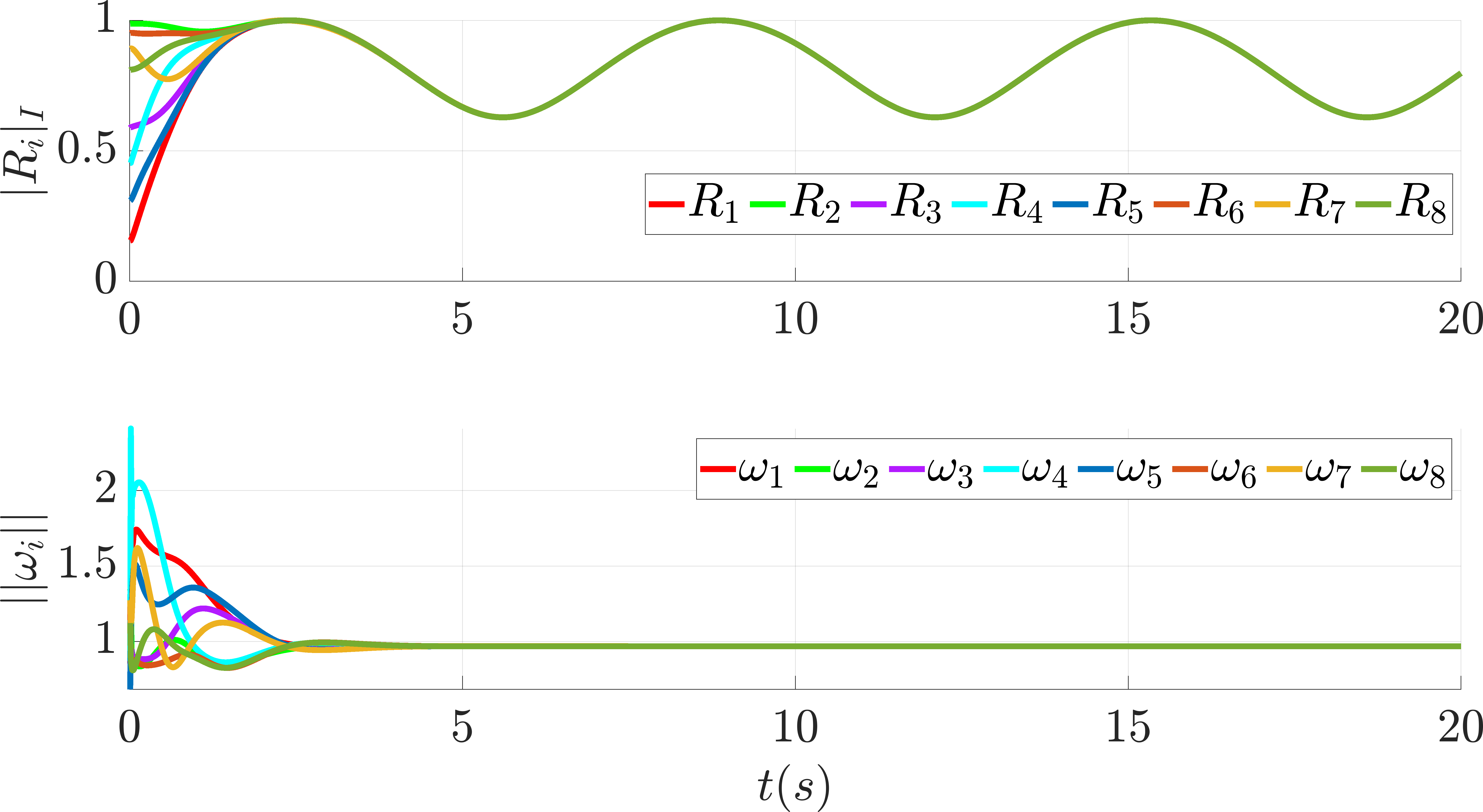}
    \caption{Time evolution of the attitudes and angular velocities under the control \eqref{torque_i}, with $k_\omega=0$ and $\bar{k}_\omega=1$. The satellites converge to a common time-varying attitude with a constant angular velocity.}
    \label{sim_wc}
\end{figure}

For the leader–follower scenario, agent~1 is designated as the leader and has access to reference vector measurements corresponding to the desired orientation $R_r = \mathcal{R}\left(\tfrac{\pi}{2},u\right)$. The same network topology, initial conditions, and inertial parameters are used. The additional parameters are chosen as $\bar{\rho}_1=1$, $\bar{\rho}_2=2$, and $k_1=1$. As highlighted in Theorems \ref{theorem_4} and \ref{theorem_5} and shown in Figures \ref{sim_leader} and \ref{sim_w0_leader}, the absolute attitude of each agent under both feedback laws \eqref{w_leader} and \eqref{torque_i_leader}, with $k_\omega = 1$ and $\bar{k}_\omega=1$, converges to the reference orientation $R_r$.

\begin{figure}[!]
    \centering
    \includegraphics[width=0.99\linewidth]{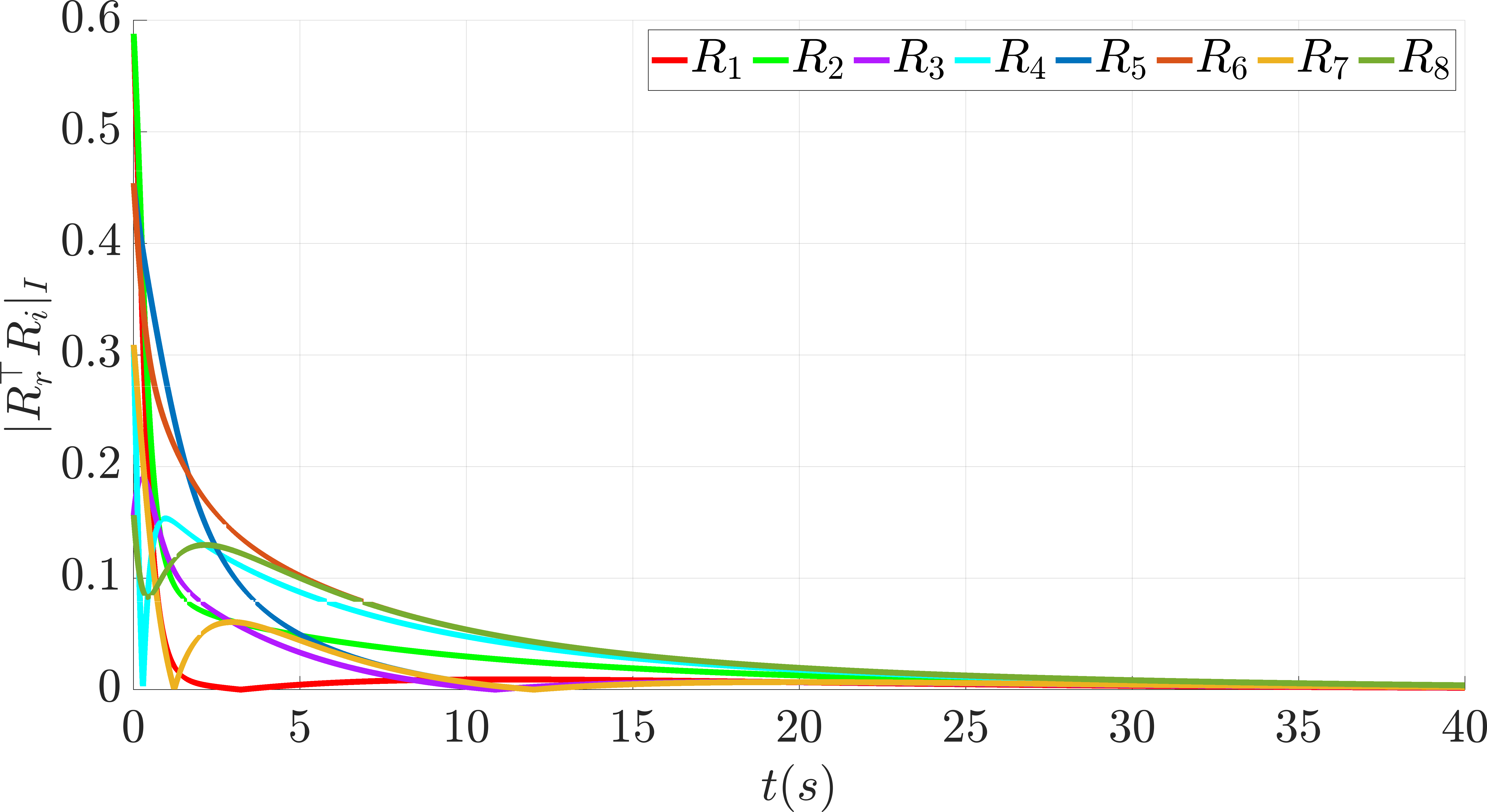}
    \caption{Time evolution of the attitudes under the control \eqref{w_leader}.}
    \label{sim_leader}
\end{figure}

\begin{figure}[h!]
    \centering
    \includegraphics[width=0.99\linewidth]{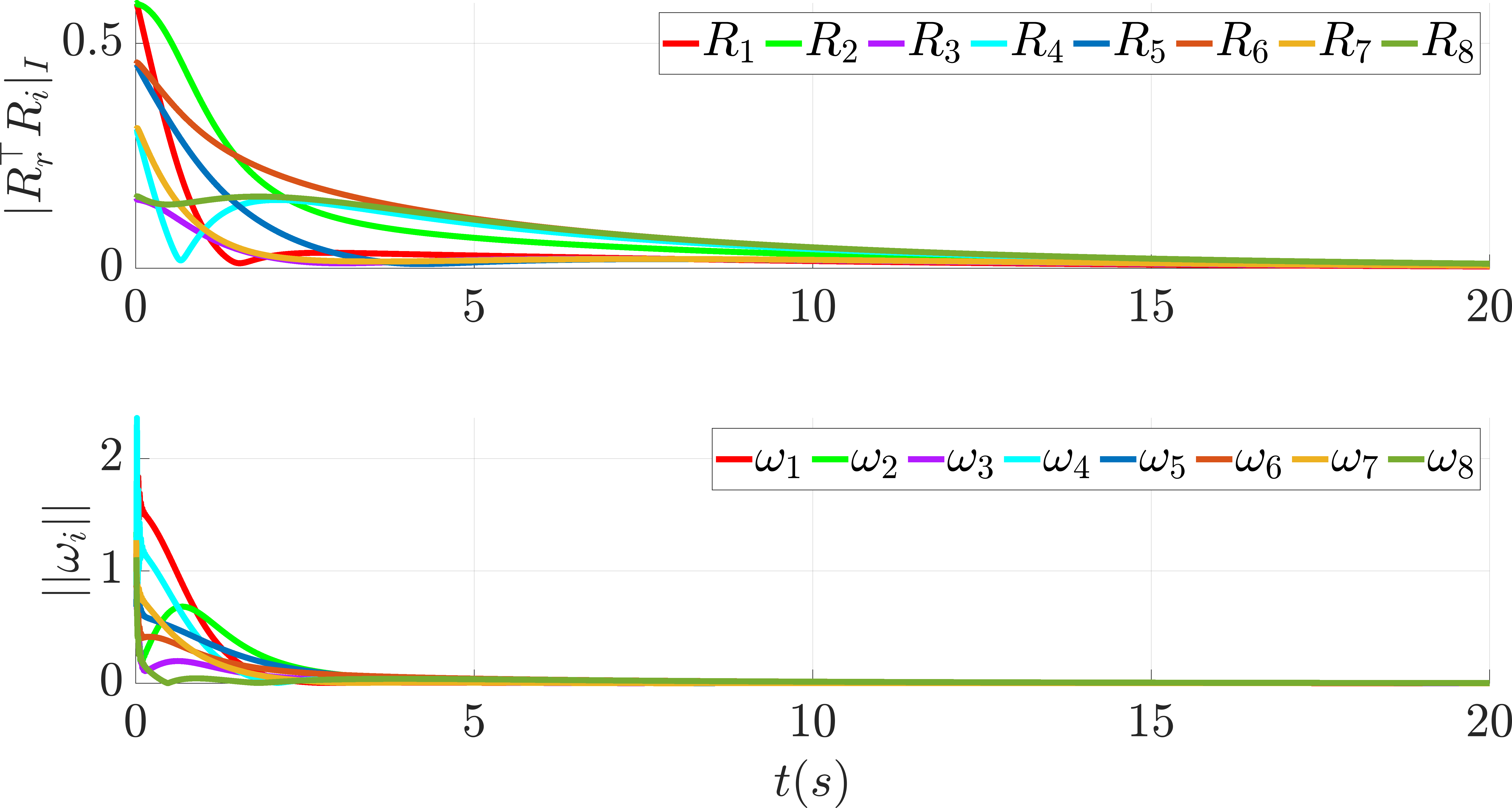}
    \caption{Time evolution of the attitudes and angular velocities under the control \eqref{torque_i_leader}, with $k_\omega=\bar{k}_\omega=k_1=1$.}
    \label{sim_w0_leader}
\end{figure}

\section{CONCLUSIONS}\label{s7}
This paper studied the distributed attitude synchronization problem for a network of rigid-body systems evolving on $SO(3)$ under partial state exchange. Relying only on local vector measurements and angular velocity information, we developed distributed synchronization schemes at both the kinematic and dynamic levels. Leaderless designs were proposed to achieve synchronization to a common unknown orientation, while leader-follower schemes were developed to ensure alignment with a prescribed constant orientation defined through reference measurements available to a single leader. All control laws were formulated directly on $SO(3)$. A rigorous stability analysis established almost global asymptotic stability of the closed-loop systems under an undirected, connected, and acyclic interaction topology. 
Future work will focus on further improving the stability properties of the proposed schemes by exploiting multi-agent hybrid control techniques, such as those recently developed in \cite{Mouaad_TAC2024,Mouaad_TAC2026}, in order to achieve global asymptotic stability and eliminate the undesired critical points inherent to smooth designs on $SO(3)$. Another important research direction is the extension of the proposed framework to more general interaction topologies, including directed graphs and time-varying networks, which would broaden the applicability of the results to a wider class of multi-agent systems.

\appendices 
\section{Proof of Theorem \ref{theorem_1}} \label{app_1}
It follows from \eqref{vector_measurement} and \eqref{fc_vector_meas} that 
{\small
\begin{align}\label{w_1}
    \omega_i 
    &=\frac{k_R}{2}R_i^\top \sum_{j \in \mathcal{N}_i} \sum_{\ell=1}^n \rho_\ell \left(R_i R_j^\top a_\ell \times a_\ell \right).
\end{align}}Using the facts that $x \times y = 2 \psi(y x^\top)$ for every $x, y \in \mathbb{R}^3$, one has 

{\small
\begin{align}
    \omega_i &= k_R R_i^\top \sum_{j \in \mathcal{N}_i} \psi(AR_j R_i^\top).
\end{align}}Recall that $A:=\sum_{l=1}^n \rho_\ell a_\ell a_\ell^\top$, and that the map $\psi$ is introduced in the notation subsection. Since the graph $\mathcal{G}$ is undirected but endowed with virtual orientation, the neighborhood $\mathcal{N}_i$ of an agent $i$ can be decomposed as $\mathcal{N}_i = \mathcal{I}_i \cup \mathcal{O}_i$, where $\mathcal{I}_i$ denotes the subset of neighbors $j \in \mathcal{N}_i$ where $j$ acts as the tail of the oriented edge $(i,j) \in \mathcal{E}$ (\ie, edges directed from $j$ to $i$) and $\mathcal{O}_i$ denotes the subset of neighbors $j \in \mathcal{N}_i$ where $j$ acts as the head of the oriented edge $(i,j) \in \mathcal{E}$ (\ie, edges directed from $i$ to $j$). Thus, $\omega_i$ can be expressed as follows:

{\small
\begin{align}
    \omega_i &=k_R R_i^\top \left(\sum_{j \in \mathcal{I}_i } \psi(AR_j R_i^\top)+\sum_{j \in \mathcal{O}_i } \psi(AR_j R_i^\top)\right)\nonumber\\
    &=k_R R_i^\top \left(\sum_{j \in \mathcal{I}_i } \psi(AR_j R_i^\top)-\sum_{j \in \mathcal{O}_i } \psi(R_i R_j^\top A)\right)\label{eq_1}\\
    &=k_R R_i^\top \left(\sum_{j \in \mathcal{I}_i } \psi(AR_j R_i^\top)-\sum_{j \in \mathcal{O}_i } R_i R_j^\top \psi(AR_i R_j^\top)\right)\label{eq_2}\\
    &=k_R R_i^\top \left(\sum_{n \in \mathcal{M}_i^+} \psi(A \bar R_n)-\sum_{l \in \mathcal{M}_i^- } \bar R_l \psi(A \bar R_l)\right)\nonumber\\
    &=k_R R_i^\top \sum_{k=1}^M H_{ik} \psi(A \bar R_k),
\end{align}
}where $H_{ik}$ is given in (\ref{H_bar}). Equations (\ref{eq_1}) and (\ref{eq_2}) are obtained using the facts that $\psi(BR)=-\psi(R^\top B)$ and $\psi(GR)=R^\top \psi(RG)$, $\forall G, B=B^\top \in \mathbb{R}^{3\times3} $ and $R\in SO(3)$. Moreover, one can verify that
 \begin{equation}\label{w_c}
     \omega = k_R \mathbf{R}^\top \mathbf{H} \Psi,
 \end{equation}
 where $\Psi:=\left[\psi(A\bar R_1)^\top, \psi(A\bar R_2)^\top, \hdots, \psi(A\bar R_M)^\top\right]^\top \in \mathbb{R}^{3M}$. The matrices $\mathbf{R}$ and $\mathbf{H}$ are introduced in \eqref{w_bar}. Next, define the following Lyapunov function candidate:
 {\small
\begin{equation}
    V(x) = \sum_{k=1}^{M} \text{tr}\left(A(I_3-\bar{R}_k)\right),
\end{equation}}which is positive definite on $\mathcal{S}$ with respect to $\mathcal{A}$. It follows from the dynamics \eqref{R_bar_dynamics_k} that $\dot{V}(x)=-\sum_{k=1}^{M} \text{tr}\left(A\bar{R}_k [\bar \omega_k]^\times\right)=2\bar \omega^\top \Psi$. We have used the facts $\text{tr}\left(B[x]^\times\right)=\text{tr}\left(\mathbb{P}_a(B)[x]^\times\right)$ and $\text{tr}\left([x]^\times [y]^\times\right)=-2x^\top y$, $\forall x, y\in \mathbb{R}^3$ and $\forall B \in \mathbb{R}^{3\times3}$ to get the last equation.
From \eqref{w_bar}, one has 
\begin{equation}
    \dot V(x)=-2 \omega^\top \mathbf{R}^\top \mathbf{H} \Psi.
\end{equation}
Furthermore, it follows from \eqref{w_c} that 
{\small
\begin{equation}\label{V_dot1}
    \dot V(x)=-2 k_R ||\mathbf{R}^\top \mathbf{H} \Psi||^2 \leq 0.
\end{equation}}Therefore, the equilibrium set $\mathcal{A}$ is stable. Furthermore, as per LaSalle’s invariance theorem, all solutions $x$ of system \eqref{R_bar_dynamics_k} with \eqref{fc_vector_meas} must converge to the largest invariant set within the region where $\dot{V}(x) = 0$, \ie, where $\textit{\textbf{H}}\Psi = 0$. By \cite[Lemma 2]{Mouaad_ACC23}, the condition $\textit{\textbf{H}}\Psi = 0$ requires $\Psi = 0$. This result further yields the equality $A \bar{R}_k = \bar{R}_k^\top  A$, for all $k \in \mathcal{M}$. Now, considering similar arguments as in \cite[Lemma 2]{Mayhew_ACC2011}, one can show that every solution \(x\) of the dynamics \eqref{R_bar_dynamics_k} with \eqref{fc_vector_meas} must converge to the set \(\Upsilon\). This concludes the proof of item \eqref{set_of_equilibrium}.

Next, we will demonstrate the almost global asymptotic stability of $\mathcal{A}$ by showing that the stable manifold associated with the undesired equilibrium set $\Upsilon \setminus \mathcal{A}$ has zero Lebesgue measure. To do this, we derive the Jacobian matrix associated with the dynamics \eqref{R_bar_dynamics_k} under the feedback control \eqref{fc_vector_meas} at the undesired equilibria, and show that this Jacobian matrix has no imaginary eigenvalues and has at least one positive eigenvalue. Let $\bar R_k=\bar R_k^* \exp{\left([\bar r_k]^\times\right)}$,  where $\bar r_k\in \mathbb{R}^3$ is sufficiently small and $\bar R_k^* \in \{\bar R_n, \bar R_m\}$. Considering the expression of $\bar R_k$ and using the first-order approximation $\exp\big([\bar r_k]^\times\big) \approx I_3 + [\bar r_k]^\times$ for sufficiently small $\bar r_k$, we obtain the following linearization of $\bar R_k$ around any point in the undesired equilibrium set $\Upsilon \setminus \mathcal{A}$:
\begin{equation}\label{first_order_app}
    \bar R_k = \bar R_k^*(I_3 + [\bar r_k]^\times),
\end{equation}
for every $k \in \mathcal{M}$. Define $\bar r:=\left[\bar r_1^\top, \bar r_2^\top,\dots, \bar r_M^\top\right]^\top \in \mathbb{R}^{3M}$. From dynamics \eqref{R_bar_dynamics_k} under feedback control \eqref{fc_vector_meas}, one obtains the following linearized system:
\begin{equation}\label{linear_sys_und}
    \dot{\bar r}=-\frac{k_R}{2} \mathbf{H}_c^\top \mathbf{H}_c\, \mathbf{A}^*\, \bar r,
\end{equation}
where $\mathbf{A}^* = \text{diag}(A^*_1, A^*_2, \hdots, A^*_M)\in \mathbb{R}^{3M\times 3M}$ with $A^*_k=\text{tr}(A\bar R^*_k)I_3 - A \bar R^*_k$ and  
\begin{equation}\label{H_bar_c}
   \mathbf{H}_c:=[H^c_{ik}]_{N\times M} \hspace{0.3cm} \text{with} \hspace{0.3cm} H^c_{ik}=\begin{cases}
      I_3 & k\in\mathcal{M}_i^+\\
      -\bar R^*_k & k\in\mathcal{M}_i^-\\
      0 & \text{otherwise}
    \end{cases}.   
\end{equation}
Linear dynamics \eqref{linear_sys_und} has been obtained by using the fact that $B [z]^\times + [z]^\times B^{\top} = \left[\left(\mathrm{tr}(B)I_3-B^{\top}\right) z\right]^\times$, $ \forall z\in \mathbb{R}^3$, $\forall B \in \mathbb{R}^{3\times3}$. The matrix $-\frac{k_R}{2} \mathbf{H}_c^\top \mathbf{H}_c\, \mathbf{A}^*$ represents the Jacobian matrix at the undesired equilibria, where $\mathbf{H}_c$ is a constant matrix. Note that the eigenvalues of the matrix $\mathbf{A}^*$ are given by the union of the eigenvalues of $A^*_k$ for all $k \in \mathcal{M}$. From the fact that \((\bar{R}^*_1, \bar{R}^*_2, \dots, \bar{R}^*_M) \in \Upsilon \setminus \mathcal{A}\), one has $A_m^* = \text{tr}(A)\, I_3 - A$ and $A_n^* \in \{{}^{1}A_n, {}^{2}A_n, {}^{3}A_n\}$, for every $m \in \mathcal{M}^I$ and $n \in \mathcal{M}^\pi$, where 
${}^{\beta_n}A_n = \text{tr}\left(A \mathcal{R}(\pi, u_{\beta_n})\right)I_3 - A \mathcal{R}(\pi, u_{\beta_n})$
for each $\beta_n \in \{1, 2, 3\}$ and $u_{\beta_n} \in \mathcal{E}(A)$. Moreover, using the fact that $\mathcal{R}(\pi, u_{\beta_n})=-I_3+2u_{\beta_n} u_{\beta_n}^\top$, it follows that
${}^{\beta_n}A_n = -\text{tr}(A)I_3 + 2\lambda_{\beta_n}I_3 + A - 2\lambda_{\beta_n}u_{\beta_n}u_{\beta_n}^\top$, 
where $\lambda_{\beta_n}$ is the eigenvalue of $A$ corresponding to the eigenvector $u_{\beta_n}$. Note that, for each \(n \in \mathcal{M}^\pi \), the matrix \( A_n^* \) can take one of three possible values (\ie, ${}^{1}A_n, {}^{2}A_n$ or ${}^{3}A_n$), depending on the choice of the eigenvector of the matrix \( A \). The set of eigenpairs of the matrix \( A_m^* \) is given by $\left\{(\lambda_2 + \lambda_3, u_1), (\lambda_1 + \lambda_3, u_2), (\lambda_1 + \lambda_2, u_3)\right\}$, for all \( m \in \mathcal{M}^I \). For the matrices \( {}^{1}A_n \), \( {}^{2}A_n \), and \( {}^{3}A_n \), the eigenpair sets are found to be: $\{(-\lambda_2 - \lambda_3, u_1), (\lambda_1 - \lambda_3, u_2), (\lambda_1 - \lambda_2, u_3)\}$, $\{(\lambda_2 - \lambda_3, u_1), (-\lambda_1 - \lambda_3, u_2), (\lambda_2 - \lambda_1, u_3)\}$, and $\{(\lambda_3 - \lambda_2, u_1), (\lambda_3 - \lambda_1, u_2), (-\lambda_1 - \lambda_2, u_3)\}$, respectively,
for all \( n \in \mathcal{M}^\pi \). Given the fact that the matrix $A$ is a positive semi-definite matrix with three distinct eigenvalues ( \ie, \( \lambda_1 \neq \lambda_2 \neq \lambda_3 \) ), one can check that the eigenvalues of \( A^*_k \) must either be all negative or contain a mixture of positive and negative values which implies that the eigenvalues of the matrix $\mathbf{A}^*$ are either all negative or some of them are positive and some are negative. We will now demonstrate that the matrices $\mathbf{A}^*$ and $\mathbf{H}_c^\top \mathbf{H}_c\, \mathbf{A}^*$ have the same number of positive and negative eigenvalues. Since the interaction graph is an undirected tree, and by following arguments similar to those in the proof of \cite[Lemma 2]{Mouaad_ACC23}, one can show that $\mathbf{H}_c$ has full column rank. Consequently, we can rewrite the matrix $\mathbf{H}_c^\top \mathbf{H}_c\, \mathbf{A}^*$ as:  
\begin{align}\label{similarity}
    \mathbf{H}_c^\top \mathbf{H}_c\, \mathbf{A}^* = \left(\mathbf{H}_c^\top \mathbf{H}_c\right)^\frac{1}{2}\,\mathbf{U}\,\left(\mathbf{H}_c^\top \mathbf{H}_c\right)^{-\frac{1}{2}},
\end{align}
where $\mathbf{U}:=\left(\mathbf{H}_c^\top \mathbf{H}_c\right)^\frac{1}{2}\, \mathbf{A}^*\,\left(\mathbf{H}_c^\top \mathbf{H}_c\right)^\frac{1}{2}$. From \eqref{similarity}, one can see that the matrices $\mathbf{U}$ and $\mathbf{H}_c^\top \mathbf{H}_c\, \mathbf{A}^*$ are similar, \ie, they have identical eigenvalues. Furthermore, by Sylvester's law of inertia \cite[Theorem 4.5.8]{Horn_Johnson_2012}, $\mathbf{U}$ and $\mathbf{A}^*$ have the same number of positive and negative eigenvalues. Given that the eigenvalues of $\mathbf{A}^*$ are either all negative or a mix of positive and negative, it follows that the Jacobian matrix $-\frac{k_R}{2} \mathbf{H}_c^\top \mathbf{H}_c\, \mathbf{A}^*$ has at least one positive eigenvalue and no zero eigenvalue. As a result, the set of undesired equilibrium points $\Upsilon \setminus \mathcal{A}$ is unstable. Moreover, by the \textit{stable manifold theorem} \cite{Perko_book}, 
the stable manifold associated with the undesired equilibrium set $\Upsilon \setminus \mathcal{A}$ has Lebesgue measure zero. Consequently, the equilibrium set $\mathcal{A}$ is almost globally asymptotically stable. This concludes the proof of items (\ref{unstability_of_equilibrium}) and (\ref{stability_of_equilibrium}).

\section{Proof of Theorem \ref{theorem_2}} \label{app_2}
Applying the same calculations used to derive \eqref{w_c} from \eqref{w_1}, and considering \eqref{w_dynamics_i} and \eqref{torque_i}, it can be verified that 

{\small
    \begin{equation}\label{w_total_dyn}
        \mathbf{J} \dot \omega =k_R \mathbf{R}^\top \mathbf{H} \Psi-k_\omega \omega - \bar k_\omega (\mathcal{L}\otimes I_3) \omega,
    \end{equation}}where 
    $\mathbf{J} :=\text{diag}(J_1, J_2, \hdots, J_N)\in \mathbb{R}^{3N\times 3N}$, and $\mathcal{L}:= H H^\top  \in \mathbb{R}^{N\times N}$ is the Laplacian matrix corresponding to graph $\mathcal{G}$. Consider the following Lyapunov function candidate:

    {\small
    \begin{equation}\label{pf}
        \bar V(\bar x) = k_R \sum_{k=1}^{M} \text{tr}\left(A(I_3-\bar{R}_k)\right)+\omega^\top \mathbf{J} \omega,
    \end{equation}}which is positive definite on $\bar{\mathcal{S}}$ with respect to $\bar{\mathcal{A}}_0$. The time-derivative of $\bar V$, along the trajectories of the dynamics \eqref{w_dynamics_i} and \eqref{R_bar_dynamics_k} with \eqref{torque_i}, is given by

{\small
\begin{equation}\label{v_z_dot}
    \dot{\bar V}(\bar x)=- 2 k_\omega ||\omega||^2 -2 \bar k_\omega ||(H^\top \otimes I_3) \omega||^2.
\end{equation}}Furthermore, from the fact that $R_i$ and $\omega_i$ are bounded (since $\dot{\bar V}\leq0$), one can verify that $\Ddot{\bar{V}}$ is also bounded. This implies that $\dot{\bar{V}}$ is uniformly continuous, and by Barbalat’s lemma,  all solutions \(\bar{x}\) converge to the largest invariant set contained in \(\{\bar{x} \in \bar{\mathcal{S}} \mid \dot{\bar{V}}(\bar{x}) = 0\}\). The condition $\dot{\bar{V}}(\bar{x}) = 0$ implies $\omega = 0$, which, according to \eqref{w_total_dyn}, leads to $\mathbf{H} \Psi = 0$. Consequently, based on \cite[Lemma 2]{Mouaad_ACC23} and the arguments in \cite[Lemma 2]{Mayhew_ACC2011}, it follows that any solution $\bar{x}$ of the dynamics \eqref{w_dynamics_i} and \eqref{R_bar_dynamics_k} with \eqref{torque_i} converges to the set $\bar{\Upsilon}_0 :=\bar{\mathcal{A}}_0 \cup \{x \in \mathcal{S} : \omega=0, \, \bar{R}_m = I_3, \, \bar R_n=\mathcal{R}(\pi, u_{\beta_n}), \, \forall m \in \mathcal{M}^I, \, \forall n \in \mathcal{M}^\pi\}$. This concludes the proof of item \eqref{dyn_set_of_equilibrium}.

To prove items \eqref{dyn_unstability_of_equilibrium} and \eqref{dyn_stability_of_equilibrium}, we will show that the set $\bar \Upsilon_0 \setminus \bar{\mathcal{A}}_0$, which consists of undesired equilibrium points, is unstable and that its stable manifold has zero Lebesgue measure. we will first show that the equilibrium set $\bar \Upsilon_0$ coincides with the set of critical points of the potential function $\bar V(\bar x)$. To do so, let us find the gradient of the potential function $\bar V$ with respect to $\bar{R}_k$ and $\omega_i$ for all $k \in \mathcal{M}$ and $i \in \mathcal{V}$. Let $\mathbb{O} \subset \mathbb{R}$ be an open interval containing zero in its interior. For each $k \in \mathcal{M}$ and $i \in \mathcal{V}$, we define the two smooth curves $\varphi_k: \mathbb{O} \to SO(3)$ and $\gamma_i: \mathbb{O} \to \mathbb{R}^3$ as $\varphi_k(t) = \bar{R}_k \exp\left(t [\zeta_k]^\times\right)$ and $\gamma_i(t) = \omega_i+v_i t$, where $\bar R_k \in SO(3)$, $\omega_i\in \mathbb{R}^3$, $ \zeta_k\in \mathbb{R}^3$ and $ v_i \in \mathbb{R}^3$. Define $\bar{\mathbf{x}}(t) := \left(\varphi_1(t), \varphi_2(t), \ldots, \varphi_M(t), \gamma_1(t), \gamma_2(t), \ldots, \gamma_N(t)\right) \in \bar{\mathcal{S}}$. The derivative of $\bar{V}\left(\bar{\mathbf{x}}(t)\right)$ with respect to $t$ is given by:

{\small
\begin{align}\label{hess}
    \frac{d}{dt}\bar{V}\left(\bar{\mathbf{x}}(t)\right)=&-k_R\sum_{k=1}^{M} \text{tr}\left(A\bar{R}_k\exp{\left(t [\zeta_k]^\times\right)}[\zeta_k]^\times\right)\nonumber\\
    &-2\sum_{i=1}^N v_i^\top J_i (\omega_i+v_i t).
\end{align}}At $t=0$, one has
{\small
\begin{align}\label{gradient_v}
        \left.\frac{d}{dt}\bar{V}\left(\bar{\mathbf{x}}(t)\right)\right|_{t=0}
        =&2\,k_R\sum_{k=1}^{M} \zeta_k^\top \psi(A \bar R_k)+2\sum_{i=1}^N v_i^\top J_i \omega_i\nonumber\\
        =&2 \begin{bmatrix}
            \zeta^\top&v^\top
        \end{bmatrix}
        \begin{bmatrix}
            k_R\,\Psi\\ \mathbf{J} \omega
        \end{bmatrix}.
    \end{align}}Note that {\small $\mathbf{J}\omega=\left[\left(\nabla_{\omega_1} \bar{V}\right)^\top,\left(\nabla_{\omega_2} \bar{V}\right)^\top, \hdots, \left(\nabla_{\omega_N} \bar{V}\right)^\top\right]^\top$ $\in \mathbb{R}^{3N}$} and {\small $\Psi=\Bigg[\psi\left(\bar R^\top_1\nabla_{\bar R_1} \bar{V}\right)^\top,\psi\left(\bar R^\top_2\nabla_{\bar R_2} \bar{V}\right)^\top, \hdots, $ $\psi\left(\bar R^\top_M\nabla_{\bar R_M} \bar{V}\right)^\top\Bigg]^\top$ $\in \mathbb{R}^{3M}$}, where $\nabla_{\bar R_k} \bar{V}$ and $\nabla_{\omega_i} \bar{V}$ are the gradients of $ \bar{V}$ with respect to $\bar R_k$ and $\omega_i$, respectively, according to the \textit{Riemannian} metrics $\langle \eta_1, \eta_2 \rangle_{SO(3)} = \frac{1}{2} \operatorname{tr}(\eta_1^\top \eta_2)$ and $\langle y_1, y_2 \rangle_{\mathbb{R}^3} = y_1^\top y_2$ for every $\eta_1, \eta_2\in \mathfrak{so}(3)$ and $y_1, y_2 \in \mathbb{R}^3$. The critical points of the potential function $\bar{V}$ are given by the set $\{\bar x \in \bar{\mathcal{S}} : \Psi = 0 , \omega = 0\}$. It is clear that the set of equilibria $\bar \Upsilon_0$ of the dynamics \eqref{w_dynamics_i} and \eqref{R_bar_dynamics_k} under the control torque \eqref{torque_i} coincides with the set of critical points of the potential function $\bar V(\bar x)$. Now, let us evaluate the \textit{Hessian} of \( \bar{V}(\bar{x}) \) at $ \bar x \in \bar{\Upsilon}_0 \setminus \bar{\mathcal{A}}_0$, denoted as \( \text{\textit{Hess}}\,\bar{V}(\bar{x}) \). Consider the two smooth curves $\varphi_k$ and $\gamma_i$, where $\bar R_k=\bar R^*_k$ and $\omega_i=\omega_i^*$ with $(\bar{R}^*_1, \bar{R}^*_2, \ldots, \bar{R}^*_M, \omega^*_1, \omega^*_2, \ldots, \omega^*_N) \in \bar{\Upsilon}_0 \setminus \bar{\mathcal{A}}_0$. The second derivative of $\bar V(\bar x)$ with respect to $t$ is given by: 
    
{\small \begin{align}\label{hess}
    \frac{d^2}{dt^2}\bar V(\bar x)=&-k_R\sum_{k=1}^{M} \text{tr} \left(A\bar{R}^*_k\exp{\left(t [\zeta_k]^\times\right)}\left([\zeta_k]^\times\right)^2\right)\nonumber\\
    &-k_R\sum_{k=1}^{M} \text{tr}\left(A\bar{R}^*_k\exp{\left(t [\zeta_k]^\times\right)}[\dot \zeta_k]^\times\right)\nonumber\\
    &+\sum_{i=1}^N v_i^\top J_i v_i+\sum_{i=1}^N \dot{v}_i^\top J_i (\omega^*_i+v_i t).
\end{align}}From the fact that $\mathbb{P}_a(A\bar{R}^*_k) = 0$ and $\omega_i^* = 0$ for all $(\bar{R}^*_1, \bar{R}^*_2, \ldots, \bar{R}^*_M, \omega^*_1, \omega^*_2, \ldots, \omega^*_N) \in \Upsilon_z \setminus \mathcal{A}_z$, it follows that
{\small
\begin{align}\label{hess_2}
    \left.\frac{d^2}{dt^2}\bar V(\bar x)\right|_{t=0}=&-k_R\sum_{k=1}^{M} \text{tr} \left(A\bar{R}^*_k\left([\zeta_k]^\times\right)^2\right)+\sum_{i=1}^N v_i^\top J_i v_i.
\end{align}
}Using the identity $\left([\zeta]^\times\right)^2 = -\|\zeta\|^2 I_3 + \zeta \zeta^{\top}$ and the property $\text{tr}(\zeta_1 \zeta_2^{\top}) = \zeta_1^{\top} \zeta_2$ for all $\zeta, \zeta_1, \zeta_2 \in \mathbb{R}^3$, we further simplify equation \eqref{hess_2} as follows:
\begin{align}\label{hess_3}
    \left.\frac{d^2}{dt^2}\bar V(\bar x)\right|_{t=0}\!= k_R\sum_{k=1}^{M} \zeta^{\top}_k A^*_k \zeta_k+\sum_{i=1}^N v_i^\top J_i v_i.
\end{align}
Letting $\zeta = [\zeta_1^\top, \zeta_2^\top, \hdots, \zeta_M^\top]^\top \in \mathbb{R}^{3M}$ and $v=[v_1^\top, v_2^\top, \ldots, v_N^\top]^\top \in \mathbb{R}^{3N}$, one has
\begin{align}\label{hess_4}
    \left.\frac{d^2}{dt^2}\bar V(\bar x)\right|_{t=0}\!&=\begin{bmatrix}
        \zeta\\v
    \end{bmatrix}^\top \begin{pmatrix}
        k_R\,\mathbf{A}^*&0_{3M\times3N}\\
        0_{3N\times3M}&\mathbf{J}
    \end{pmatrix}\begin{bmatrix}
        \zeta \\ v
    \end{bmatrix}.
\end{align}
Recall that $\mathbf{A}^* = \text{diag}(A^*_1, A^*_2, \hdots, A^*_M)$, where $A^*_k=\text{tr}(A\bar R^*_k)I_3 - A \bar R^*_k$. According to \cite{Mahony_book_OAMM}, it follows from \eqref{hess_4} that  
\begin{equation}\label{mtx_hess}
    \text{\textit{Hess}}\,\bar V(\bar x)=\begin{pmatrix}
        k_R\,\mathbf{A}^*&0_{3M\times3N}\\
        0_{3N\times3M}&\mathbf{J}
    \end{pmatrix},
\end{equation} 
for every $\bar x \in \bar \Upsilon_0 \setminus \bar{\mathcal{A}}_0$. Since the eigenvalues of \( \mathbf{A}^* \) are either negative or a mixture of positive and negative values, together with the fact that the matrix \( \mathbf{J} \) is positive definite, implies that the critical points of \( \bar{V}(\bar x) \) within \( \bar \Upsilon_0 \setminus \bar{\mathcal{A}}_0 \) are saddle points of \( \bar{V}(\bar x) \). Next, we will determine the stability properties of the set of undesired equilibrium points $\bar \Upsilon_0 \setminus \bar{\mathcal{A}}_0$ for the dynamics \eqref{w_dynamics_i} and \eqref{R_bar_dynamics_k} with the control torque \eqref{torque_i}. Consider the real-valued function $\bar{V}^*(\bar x): SO(3)^M \times \mathbb{R}^{3N}\rightarrow \mathbb{R}$, defined as $\bar{V}^*(\bar x) = 2\, k_R \sum_{n \in \mathcal{M}^\pi} (\lambda_{p_n} + \lambda_{d_n}) - \bar{V}(\bar x)$, where $\lambda_{p_n}$ and $\lambda_{d_n}$ are two distinct eigenvalues of the matrix $A$, \ie, $p_n, d_n \in \{1, 2, 3\}$ with $p_n \neq d_n$. Let $\bar{x}^*=(\bar{R}^*_1, \bar{R}^*_2, \ldots, \bar{R}^*_M, \omega^*_1, \omega^*_2, \ldots, \omega^*_N) \in \bar \Upsilon_0 \setminus \bar{\mathcal{A}}_0$ with $\bar{R}^*_n = \mathcal{R}(\pi, u_{l_n})$ and $\omega_i^* = 0$, for $n \in \mathcal{M}^\pi$, where $l_n \in \{1, 2, 3\}$ satisfies $l_n \neq p_n$ and $l_n \neq d_n$. Notice that $\bar{V}^*(\bar{x}^*) = 0$. Define the set ${\mathbb{B}}_r := \{ (\bar{R}_1, \bar{R}_2, \dots, \bar{R}_M, \omega_1, \omega_2, \dots, \omega_N) \in \bar{\mathcal{S}} : |\bar{R}_1^\top \bar{R}_1^*|_I + |\bar{R}_2^\top \bar{R}_2^*|_I + \dots + |\bar{R}_M^\top \bar{R}_M^*|_I+||\omega_1||+||\omega_2||+\dots+||\omega_N|| \leq  r \}$ with $r > 0$. From the fact that the set $\bar \Upsilon_0 \setminus \bar{\mathcal{A}}_0$ consists only the saddle points of $\bar{V}(\bar x)$, one can verify that the set ${\mathbb{U}} = \{ \bar x \in {\mathbb{B}}_{r} \mid \bar{V}^*(\bar x) > 0 \}$ is non-empty. Moreover, from \eqref{v_z_dot}, one has $\dot{\bar{V}}^*(\bar x) = -\dot{\bar{V}}(\bar x) > 0$ in ${\mathbb{U}}$, which implies that any trajectory starting in the set ${\mathbb{U}}$ must exit ${\mathbb{U}}$ from boundary surface of ${\mathbb{B}}_{r}$. According to \textit{Chetaev's theorem} \cite{khalil2002nonlinear}, this shows that all points in the undesired equilibrium set $\bar \Upsilon_0 \setminus \bar{\mathcal{A}}_0$ are unstable. This completes the proof of item \eqref{dyn_unstability_of_equilibrium}. 

Now, to show that the stable manifold associated with the undesired equilibrium points has zero Lebesgue measure, we derive the Jacobian matrix corresponding to dynamics \eqref{w_dynamics_i} and \eqref{R_bar_dynamics_k} under the control torque \eqref{torque_i} at the undesired equilibria. 
Let $R_i = R^*_i \exp{\left([r_i]^\times\right)}$, $R_j = R^*_j \exp{\left([r_j]^\times\right)}$, and $\omega_i = y_i$, where $r_i,\, r_j,\, r_j-r_i,\, y_i\in \mathbb{R}^3$ are sufficiently small and $\bar R_k^* = R^*_j\left(R_i^*\right)^\top \in \{\bar R_n, \bar R_m\}$ for each $\{k\}=\mathcal{M}_i^+ \cap \mathcal{M}_j^-$. It follows that $\bar R_k = R_j R_i^\top=R^*_j \exp{\left([r_j-r_i]^\times\right)} \left(R^*_j\right)^\top$. Considering the first-order approximation $\exp\big([r]^\times\big) \approx I_3 + [r]^\times$ for sufficiently small $r$, we obtain the linearization of $\bar R_k$ and $\omega_i$ around any point in the undesired equilibrium set $\Upsilon_z \setminus \mathcal{A}_z$ as $\bar R_k = \bar R_k^*(I_3 + [\bar r_k]^\times)$ and $\omega_i = y_i$, respectively, for every $\{k\}=\mathcal{M}_i^+ \cap \mathcal{M}_j^-$ and $i \in \mathcal{V}$, where $\bar r_k = R_i^*(r_j-r_i)$. Define $\bar r:=\left[\bar r_1^\top, \bar r_2^\top,\dots, \bar r_M^\top\right]^\top \in \mathbb{R}^{3M}$ and $y:=\left[y_1^\top, y_2^\top,\dots, y_N^\top\right]^\top \in \mathbb{R}^{3N}$. It follows from the dynamics \eqref{R_dynamics_i}-\eqref{w_dynamics_i} under the control torque \eqref{torque_i} that 
\begin{align}\label{linear_sys_und_w}
\begin{bmatrix}
    \dot{\bar r} \\[4pt] \dot y
\end{bmatrix} =
\mathcal{J}
\begin{bmatrix}
    \bar r \\[4pt] y
\end{bmatrix},
\end{align}
where $$\mathcal{J}:=\begin{pmatrix}
    0_{3M\times 3M} & -\mathbf{H}_c^\top\, \mathbf{R}_c\\
    \frac{k_R}{2} \mathbf{J}^{-1} \mathbf{R}_c^\top\,\mathbf{H}_c\, \mathbf{A}^* & -\mathbf{J}^{-1} \left(k_\omega I_{3N}+\bar k_\omega (\mathcal{L}\otimes I_3)\right) 
\end{pmatrix}$$ with $\mathbf{R}_c :=\text{diag}(R^*_1, R^*_2, \hdots, R^*_N) \in \mathbb{R}^{3N \times 3N}$. Note that $\mathbf{R}_c$ is constant matrix. Moreover, from the facts that $\mathbf{J}$ is positive definite, $\mathbf{A}^*$ is invertible (its eigenvalues are nonzero), $\mathbf{H}_c$ has full column rank, and $\mathcal{L}$ is positive semidefinite, one can verify that 

{\small \begin{align}\label{det_M}
  &\text{det}(\mathcal{J}) = \frac{k_R}{2}\det(\mathbf{J}^{-1}) 
   \text{det}\left( k_\omega I_{3N} + \bar k_\omega (\mathcal{L}\otimes I_3)\right) \nonumber\\
  & \times 
   \text{det}\left(\mathbf{H}_c^\top \mathbf{R}_c \bigl(k_\omega I_{3N} + \bar k_\omega (\mathcal{L}\otimes I_3)\bigr)^{-1}\mathbf{R}_c^\top\mathbf{H}_c\right)
   \text{det}(\mathbf{A}^*) \neq 0.\nonumber
\end{align}}This implies that $\mathcal{J}$ does not have an eigenvalue at the origin. Now, we use a proof by contradiction to show that the matrix $\mathcal{J}$ does not have any eigenvalues on the imaginary axis. Assume that $\mathcal{J}$ has an imaginary eigenvalue $i\lambda$, where $i^2=-1$ and $\lambda \in \mathbb{R}\setminus\{0\}$, with a corresponding nonzero eigenvector
\[
z=\begin{pmatrix}z_1\\z_2\end{pmatrix}, \quad z_1\in\mathbb{R}^{3M},\; z_2\in\mathbb{R}^{3N},
\]
such that $\mathcal{J}z=i\lambda z$. It follows from \eqref{linear_sys_und} that 
\begin{align}
&-\mathbf{H}_c^\top\, \mathbf{R}_c\, z_2 = i\lambda\, z_1, \label{sys_equ_1}\\
\frac{k_R}{2} \mathbf{J}^{-1}\mathbf{R}_c^\top \mathbf{H}_c \mathbf{A}^* z_1 -&\mathbf{J}^{-1} \left(k_\omega I_{3N}+\bar k_\omega (\mathcal{L}\otimes I_3)\right) z_2 = i\lambda z_2. \label{sys_equ_2}
\end{align}
Equation \eqref{sys_equ_1} can be rewritten as follows: 
\begin{equation}\label{sys_equ_3}
    z_1 = \frac{i}{\lambda} \mathbf{H}_c^\top \mathbf{R}_c\, z_2.
\end{equation}
From \eqref{sys_equ_2} and \eqref{sys_equ_3}, one has 
\begin{align}\label{sys_equ_4}
     i\lambda z_2&=\frac{ik_R}{2\lambda} \mathbf{J}^{-1}\mathbf{R}_c^\top \mathbf{H}_c \mathbf{A}^*\mathbf H_c^\top \mathbf R_c z_2\nonumber\\
     &~~~~~~~~~~~~~~~~~~-\mathbf{J}^{-1} \left(k_\omega I_{3N}+\bar k_\omega (\mathcal{L}\otimes I_3)\right) z_2.
\end{align}
Furthermore, by multiplying \eqref{sys_equ_4} on the left by $z_2^\top \mathbf{J}$ and rearranging it, one obtains:
\begin{align}\label{sys_equ_5}
    &\left(\frac{k_R}{2\lambda} z_2^\top\mathbf{R}_c^\top\mathbf{H}_c \mathbf{A}^*\mathbf H_c^\top \mathbf R_c z_2-\lambda z_2^\top \mathbf{J} z_2 \right)i\nonumber\\
    &~~~~~~~~~~~~~~~~~-z_2^\top \left(k_\omega I_{3N}+\bar k_\omega (\mathcal{L}\otimes I_3)\right) z_2 =0.
\end{align}
It follows from the real part of \eqref{sys_equ_5} that
\begin{align}
    k_\omega ||z_2||^2 +\bar k_\omega ||(H^\top\otimes I_3) z_2||^2 =0,
\end{align}
which implies that $z_2 = 0$. This further implies that $z_1 = 0$ according to \eqref{sys_equ_3}. Hence, the corresponding eigenvector is the zero vector, which is a contradiction. This means that the matrix $\mathcal{J}$ does not have any eigenvalues on the imaginary axis. Moreover, since all points in the undesired equilibrium set $\bar \Upsilon_0 \setminus \bar{\mathcal{A}}_0$ are unstable and the linearized system \eqref{linear_sys_und} does not have an eigenvalue at the origin, the matrix $\mathcal{J}$ must have at least one eigenvalue with a positive real part. Therefore, by the \textit{stable manifold theorem},
the stable manifold associated with the undesired equilibrium set $\bar \Upsilon_0 \setminus \bar{\mathcal{A}}_0$ has zero Lebesgue measure. Thus, the equilibrium set $\mathcal{A}_0$ is almost globally asymptotically stable. This completes the proof of item \eqref{dyn_stability_of_equilibrium}.

\section{Proof of Proposition \ref{theorem_3}} \label{app_3}

Before proceeding with the proof of item \eqref{dyn_c_set_of_equilibrium}, we first show that $(\mathbf{1}_N \otimes I_3)^\top \mathbf{R}^\top \mathbf{H} \Psi = 0$, since this identity will be used in the subsequent calculations. From derivations presented in \eqref{w_1}-\eqref{w_c}, one can verify that

{\small
\begin{align}\label{eqq2}
    (\mathbf{1}_N \otimes I_3)^\top \mathbf{R}^\top \mathbf{H} \Psi &= \frac{1}{2} \sum_{i=1}^N \sum_{j \in \mathcal{N}_i} \sum_{l=1}^n \rho_\ell \left(R_j^\top a_\ell \times R_i^\top a_\ell \right).
\end{align}}Using the facts that $x \times y = 2 \psi(y x^\top)$ for every $x, y \in \mathbb{R}^3$, equation \eqref{eqq2} can be rewrite as follows

{\small
\begin{align}
    (\mathbf{1}_N \otimes I_3)^\top \mathbf{R}^\top \mathbf{H} \Psi &=\sum_{i=1}^N \sum_{j \in \mathcal{N}_i}\psi \left(R_i^\top A R_j\right).
\end{align}}Recall that $A=\sum_{\ell=1}^n \rho_\ell a_\ell a_\ell^\top$. Furthermore, one has
 \begin{align}
    (\mathbf{1}_N \otimes I_3)^\top \mathbf{R}^\top \mathbf{H} \Psi &=\sum_{i=1}^N \sum_{j=1}^N d_{ij}  \psi \left(R_i^\top A R_j\right), \label{eqq4}
\end{align}
where $d_{ij}$ is the $(i, j)$ entry of the adjacency matrix $D$ corresponding to the graph $\mathcal{G}$. It follows from \eqref{eqq4} that 
{\small
\begin{align}\label{eqq5}
     &(\mathbf{1}_N \otimes I_3)^\top \mathbf{R}^\top \mathbf{H} \Psi\nonumber\\
     &~~=\frac{1}{2}\left(\sum_{i=1}^N \sum_{j=1}^N d_{ij} \psi \left(R_i^\top A R_j\right)+\sum_{i=1}^N \sum_{j=1}^N d_{ij} \psi \left(R_i^\top A R_j\right) \right)\nonumber\\
     &~~=\frac{1}{2}\left(\sum_{i=1}^N \sum_{j=1}^N d_{ij} \psi \left(R_i^\top A R_j\right)-\sum_{i=1}^N \sum_{j=1}^N d_{ij} \psi \left(R_j^\top A R_i\right) \right).
 \end{align}
 }The fact that $\psi(B) = -\psi(B^\top)$, for every $B \in \mathbb{R}^{3 \times 3}$, has been used to derive the last equation. Since the graph $\mathcal{G}$ is undirected, one checks that $d_{ij} = d_{ji}$. Taking this into account and rearranging the order of summation indices in the second part of equation \eqref{eqq5}, we obtain
 {\small
 \begin{align*}
     &(\mathbf{1}_N \otimes I_3)^\top \mathbf{R}^\top \mathbf{H} \Psi\\
     &=\frac{1}{2}\left(\sum_{i=1}^N \sum_{j=1}^N d_{ij} \psi \left(R_i^\top A R_j\right)-\sum_{j=1}^N \sum_{i=1}^N d_{ji} \psi \left(R_j^\top A R_i\right) \right)\nonumber\\
     &=\frac{1}{2}\left(\sum_{i=1}^N \sum_{j=1}^N d_{ij} \psi \left(R_i^\top A R_j\right)-\sum_{i=1}^N \sum_{j=1}^N d_{ij} \psi \left(R_i^\top A R_j\right) \right)=0. \nonumber
 \end{align*}
 }Now, consider the potential function $\bar V$ defined in \eqref{pf}. The time-derivative of $\bar{V}$, considering the dynamics \eqref{w_dynamics_i} and \eqref{R_bar_dynamics_k} with $k_\omega = 0$ and $\bar{k}_\omega >0$, is 
$
    \dot{{\bar V}} = -2 \bar k_\omega ||(H^\top \otimes I_3) \omega||^2.
$
From the fact that $\dot{{\bar{V}}}\leq 0$, it follows that $\omega_i$ is bounded, and consequently, $\Ddot{\bar{{V}}}$ is also bounded. This, in turn, implies that $\dot{\bar{V}}$ is uniformly continuous. Thus, by Barbalat’s lemma, any solution $\bar x$ to the closed-loop system \eqref{w_dynamics_i} and \eqref{R_bar_dynamics_k} must converge to the largest invariant set contained in the set characterized by \(\{\bar{x} \in \bar{\mathcal{S}} \mid \dot{\bar{ V}}(\bar x)=0\}\). The condition  $\dot{\bar{V}}(\bar x)=0$ implies that $\omega = \mathbf{1}_N \otimes \Omega$, which in turn leads to $\dot{\omega} = \mathbf{1}_N \otimes \dot{\Omega}$. Considering this fact, it follows from \eqref{w_total_dyn} that
$$
    \mathbf{J} \left(\mathbf{1}_N \otimes \dot{\Omega} \right) =k_R \mathbf{R}^\top \mathbf{H} \Psi.
$$
Multiply both sides by $(\mathbf{1}_N \otimes I_3)^\top$ yields:
\begin{equation}\label{w_total_dyn_LaS_1}
   (\mathbf{1}_N \otimes I_3)^\top \mathbf{J} \left(\mathbf{1}_N \otimes \dot{\Omega} \right) =k_R (\mathbf{1}_N \otimes I_3)^\top \mathbf{R}^\top \mathbf{H} \Psi.
\end{equation}
Since $(\mathbf{1}_N \otimes I_3)^\top \mathbf{R}^\top \mathbf{H} \Psi = 0$, it follows from \eqref{w_total_dyn_LaS_1} that 
\begin{equation}\label{w_total_dyn_LaS_2}
   (\mathbf{1}_N \otimes I_3)^\top \mathbf{J} \left(\mathbf{1}_N \otimes \dot{\Omega} \right) = \left( \sum_{i=1}^N J_i\right) \dot \Omega=0.
\end{equation}
Since the matrix \( J_i \) is positive definite for every \( i \in \mathcal{V} \), it follows from \eqref{w_total_dyn_LaS_2} that \( \dot{\Omega} = 0 \). 
Thus, going back to Equation \eqref{w_total_dyn}, one has 
 \begin{equation}\label{eqq7}
     \mathbf{R}^\top \mathbf{H} \Psi = 0.
 \end{equation}
 Again, considering \cite[Lemma 2]{Mouaad_ACC23} and the arguments in \cite[Lemma 2]{Mayhew_ACC2011}, equation \eqref{eqq7} implies that any solution $\bar{x}$ of the dynamics \eqref{w_dynamics_i} and \eqref{R_bar_dynamics_k}, with \eqref{torque_i} for $k_\omega = 0$ and $\bar{k}_\omega >0$, converges to the set $\bar{\Upsilon}_c$. This completes the proof of item \eqref{dyn_c_set_of_equilibrium}.

To prove item \eqref{converge_item}, we need to establish the existence of a positive scalar $c>0$ used in inequality \eqref{initial_cond}. For every $R=\mathcal{R}(\pi,u)$, with $u\in\mathbb{S}^2$, one has
\begin{align}\label{tr_ineq}
    4\,\lambda_{\min}^{A_c} \leq \mathrm{tr}\!\left(A(I-R)\right) \leq 4\,\lambda_{\max}^{A_c},
\end{align}
where $A_c=\tfrac{1}{2}(\mathrm{tr}(A)I_3-A)$, $\lambda_{\min}^{A_c}$ and $\lambda_{\max}^{A_c}$ denote the minimum and maximum eigenvalue of $A_c$, respectively. One can verify that the matrix $A_c$ is positive definite, \ie, $\lambda_{\min}^{A_c}>0$. It follows from \eqref{tr_ineq} that $\mathrm{tr}\!\left(A(I-R)\right) < 4\,\lambda_{\min}^{A_c}$ for any $R\neq\mathcal{R}(\pi,u)$. Let $c = 4\,k_R\,\lambda_{\min}^{A_c}$. Since $\dot{\bar V}\leq 0$, it follows that
\begin{align}\label{v_ineq}
    \bar V(\bar x(t)) < 4\,k_R\,\lambda_{\min}^{A_c},
\end{align}
for every $\bar x(0)$ satisfies inequality \eqref{initial_cond} and all $t\geq0$. Consequently, \eqref{v_ineq} implies that
\begin{equation}
    \mathrm{tr}\!\left(A(I-\bar R_k(t))\right) < 4\,\lambda_{\min}^{A_c},
\end{equation}
for every $k\in\mathcal{M}$ and all $t\geq0$. Therefore, $\bar R_k(t)\neq \mathcal{R}(\pi,u)$ for all $t\geq0$, and thus $\bar R_k(t)$ must converge to the identity. Hence, if $\bar x(0)$ satisfies inequality \eqref{initial_cond} with $c=4\,k_R\,\lambda_{\min}^{A_c}$, the solution of \eqref{w_dynamics_i} and \eqref{R_bar_dynamics_k} under the control law \eqref{torque_i} converges to the desired set $\bar{\mathcal{A}}_c$. This completes the proof of item \eqref{converge_item}.

\section{Proof of Theorem \ref{theorem_4}} \label{app_4}
Similar to the derivative calculations employed at the beginning of the proof of Theorem~\ref{theorem_1}, one obtains
\begin{equation}\label{w_leader_c}
     \omega = k_R \mathbf{R}^\top \mathbf{H} \Psi + k_1\, \mathbf{R}^\top \, \Psi_1,
 \end{equation}
where $\Psi_1:=\left[\psi(\bar A \tilde{R}_1)^\top~~~0_{1\times3(N-1)}\right]^\top$. Now, consider the following Lyapunov function candidate:
\begin{equation}
    V^r(x^r) =k_R\, \sum_{k=1}^{M} \text{tr}\left(A(I_3-\bar{R}_k)\right)+k_1\, \text{tr}(\bar A(I_3-\tilde{R}_1)),
\end{equation}
which is positive definite on $\mathcal{S}^r$ with respect to $\mathcal{A}^r$. From the dynamics \eqref{R_bar_dynamics_k} and \eqref{R1_tilde}, and under the control law \eqref{w_leader}, the time derivative of the Lyapunov function $V^r(x^r)$ along the system trajectories is given by
\begin{align}
    \dot{V}^r(x^r)&=2k_R\,\bar \omega^\top \Psi-2k_1\,\omega^\top\mathbf{R}^\top\Psi_1.
\end{align} 
 In deriving the last equality, we used the identities$\text{tr}\left(B[x]^\times\right)=\text{tr}\left(\mathbb{P}_a(B)[x]^\times\right)$ and $\text{tr}\left([x]^\times [y]^\times\right)=-2x^\top y$, $\forall x, y\in \mathbb{R}^3$ and $\forall B \in \mathbb{R}^{3\times3}$. In view of \eqref{w_bar}, one has 
\begin{align}
    \dot{V}^r(x^r)&=-2 \omega^\top \left(k_R\,\mathbf{R}^\top \mathbf{H}\, \Psi +k_1\,\mathbf{R}^\top\Psi_1\right).
\end{align}
Under the control law \eqref{w_leader}, it follows that 
\begin{equation}\label{V_r_dot1}
    \dot V(x)=-2\, ||k_R\,\mathbf{R}^\top \mathbf{H}\, \Psi +k_1\,\mathbf{R}^\top\Psi_1||^2 \leq 0.
\end{equation}
Therefore, the equilibrium set \(\mathcal{A}^r\) is stable. Moreover, by LaSalle’s invariance principle, all solutions \(x^r\) of the systems
\eqref{R_bar_dynamics_k} and \eqref{R1_tilde} under \eqref{fc_vector_meas} converge to the largest invariant
set contained in $\left\{ x^r\in \mathcal{S}^r \;:\; \dot{V}^r(x^r)=0 \right\}$, which is characterized by $k_R\,\mathbf{H}\,\Psi + k_1\,\Psi_1 = 0$. As in \cite{li2026}, let \(\mathbf{H}=\left[\mathbf{H}_1^\top \;\; \mathbf{H}_2^\top\right]^\top\), where
\(\mathbf{H}_1 \in \mathbb{R}^{3\times 3M}\) and
\(\mathbf{H}_2 \in \mathbb{R}^{3(N-1)\times 3M}\). Then, the condition
\(\dot{V}^r(x^r)=0\) implies
\begin{align}
    k_R\, \mathbf{H}_1 \,\Psi + k_1\, \psi(\bar A \tilde{R}_1) &= 0, \label{eq_r_1}\\
    k_R\, \mathbf{H}_2 \,\Psi &= 0. \label{eq_r_2}
\end{align}
Since \(\mathbf{H}_2\) is nonsingular for all \(t \geq 0\) as shown in \cite{li2026}, it follows from \eqref{eq_r_1}--\eqref{eq_r_2} that $\Psi = 0$ and $\psi(\bar A \tilde{R}_1)=0$. Furthermore, from \cite[Lemma 2]{Mayhew_ACC2011}, one can show that every solution \(x^r\) of the dynamics \eqref{R_bar_dynamics_k} and \eqref{R1_tilde}, under the control law \eqref{w_leader}, must converge to the set \(\Upsilon^r\). This concludes the proof of item \eqref{set_of_equilibrium_leader}.

To prove items \eqref{unstability_of_equilibrium_leader} and \eqref{stability_of_equilibrium_leader}, we follow an argument analogous to that used in the proof of Theorem~\ref{theorem_1}. Specifically, We derive the Jacobian matrix associated with the dynamics given in \eqref{R_bar_dynamics_k} and \eqref{R1_tilde} under the feedback control law \eqref{w_leader}, evaluated at the undesired equilibria. Let $\bar{R}_k = \bar{R}_k^* \exp\left([\bar{r}_k]^\times\right)$ and $\tilde{R}_1 = \tilde{R}_1^* \exp\left([\tilde{r}_1]^\times\right)$, where $\bar{r}_k, \tilde{r}_1 \in \mathbb{R}^3$ are sufficiently small, $\bar{R}_k^* \in \{\bar{R}_n, \bar{R}_m\}$, and $\tilde{R}_1^* \in \{I_3, \mathcal{R}(\pi, \bar{u})\}$. Using the first-order approximation $\exp\big([z]^\times\big) \approx I_3 + [z]^\times$ for sufficiently small $z$, one obtains the following linearized system
 \begin{align}\label{linear_sys_und_leader}
\begin{bmatrix}
    \dot{\bar r} \\[4pt] \dot{\tilde r}_1
\end{bmatrix} =
\bar{\mathcal{J}}
\begin{bmatrix}
    \bar r \\[4pt] \tilde r_1
\end{bmatrix},
\end{align}
where $\bar{\mathcal{J}}:=\begin{pmatrix}
         -\frac{k_R}{2}\, \mathbf{H}_c^\top \mathbf{H}_c\, \mathbf{A}^* & 0_{3M\times3}\\
         k_R\, \mathbf{H}_{c}^1\, \mathbf{A}^* & k_1\, \tilde{A}_1
\end{pmatrix}$, 
with $\tilde{A}^*_1 = \text{tr}(\bar{A} \tilde{R}^*_1) I_3 - \bar{A} \tilde{R}^*_1$ and $\mathbf{H}_{c}^1 \in \mathbb{R}^{3\times 3M}$ is the submatrix consisting of the first three rows of $\mathbf{H}_c$. The matrix $\bar{\mathcal{J}}$ represents the Jacobian matrix associated with the dynamics given in \eqref{R_bar_dynamics_k} and \eqref{R1_tilde} under the feedback control law \eqref{w_leader}, evaluated at the undesired equilibria. The matrix $\bar{\mathcal{J}}$ is a lower triangular matrix, hence, the eigenvalues of $\bar{\mathcal{J}}$ are the union of the eigenvalues of $-\frac{k_R}{2}\, \mathbf{H}_c^\top \mathbf{H}_c\, \mathbf{A}^*$ and $k_1\, \tilde{A}_1$. Note that since $\bar{A}$ is a positive semi-definite matrix with three distinct eigenvalues, $\tilde{A}^*_1$ has strictly positive eigenvalues when $\tilde{R}_1 = I_3$. However, when $\tilde{R}_1 = \mathcal{R}(\pi, \bar{u})$, the eigenvalues of $\tilde{A}^*_1$ can either be all negative, or a mix of both positive and negative values. From this and the fact that $\mathbf{H}_c^\top \mathbf{H}_c\, \mathbf{A}^*$ and $\mathbf{A}^*$ have the same number of positive and negative eigenvalues as shown in the proof of Theorem \ref{theorem_1}, one can verify that at the undesired equilibria, the matrix $\bar{\mathcal{J}}$ has at least one positive eigenvalue and no zero eigenvalue. Consequently, the set of all undesired equilibrium points $\Upsilon^r \setminus \mathcal{A}^r$ is unstable. Furthermore, by the \textit{stable manifold theorem},
the stable manifold associated with $\Upsilon^r \setminus \mathcal{A}^r$ has zero Lebesgue measure. Hence, the equilibrium set $\mathcal{A}^r$ is almost globally asymptotically stable.  

\section{Proof of Theorem \ref{theorem_5}} \label{app_5}
Considering the proposed torque \eqref{torque_i_leader}, the angular velocity dynamics can be rewritten in the following stacked form:
\begin{equation}\label{w_total_dyn_leader}
    \mathbf{J} \dot \omega =k_R \mathbf{R}^\top \mathbf{H} \Psi+ k_1\, \mathbf{R}^\top \, \Psi_1-k_\omega \omega - \bar k_\omega (\mathcal{L}\otimes I_3) \omega.
\end{equation}
To perform the stability analysis, we consider the following candidate Lyapunov function:
\begin{equation}\label{pf_leader}
    \bar V^r(\bar x^r) = k_R \sum_{k=1}^{M} \text{tr}\left(A(I_3-\bar{R}_k)\right)+k_1\, \text{tr}(\bar A(I_3-\tilde{R}_1))+\omega^\top \mathbf{J} \omega,\nonumber
\end{equation}
whose time derivative is computed as:
\begin{equation}\label{v_z_dot_leader}
    \dot{\bar V}^r(\bar x^r)=- 2 k_\omega ||\omega||^2 -2 \bar k_\omega ||(H^\top \otimes I_3) \omega||^2.
\end{equation}
One can verify that $\ddot{\bar V}^r$ is bounded, which implies that $\dot{\bar V}^r(\bar x^r)$ is uniformly continuous. Hence, by Barbalat’s lemma, one can verify that all solutions $\bar x^r$ converge to the largest invariant set contained in $\left\{ \bar x^r \in \bar{\mathcal S}^r \,:\, \dot{\bar V}^r(\bar x^r)=0 \right\}$. Furthermore, by chasing the signals, it can be shown that this set exactly equals $\bar{\Upsilon}_0^r$. Now, consider the smooth curve $\tilde{\varphi}_1: \mathbb{O} \to SO(3)$ defined as $\tilde{\varphi}_1(t) = \tilde{R}^*_1 \exp\left(t [\tilde{\zeta}_1]^\times\right)$, where $\tilde{R}^*_1 \in SO(3)$ and $ \tilde \zeta_1\in \mathbb{R}^3$. Let $ \bar{\mathbf{x}}^r(t) := \left(\varphi_1(t), \varphi_2(t), \ldots, \varphi_M(t), \tilde{\varphi}_1(t), \gamma_1(t), \gamma_2(t), \ldots, \gamma_N(t)\right) \in \bar{\mathcal{S}}^r$, where the smooth curves $\varphi_k(t)$ and $\gamma_i(t)$ are defined in the proof of Theorem~\ref{theorem_2}. By calculating the first and second derivatives of the potential function \(\bar{V}^r(\bar{x}^r)\) with respect to time at \(t=0\), the gradient and Hessian of \(\bar{V}^r(\bar{x}^r)\) are obtained as follows:
\begin{align}
    \nabla_{\bar x^r} \bar V^r(\bar x^r) &= \begin{pmatrix}
        k_R\, \Psi\\
        k_1\, \Psi_1\\
        \mathbf{J}\, \omega
    \end{pmatrix} \label{grad_dyn_leader}\\
    \textit{Hess}\bar V^r(\bar x^r)&=\begin{pmatrix}
        k_R\,\mathbf{A}^* & 0_{3M\times3} & 0_{3M\times3N}\\
        0_{3\times 3M} & k_1\, \tilde{A}_1 & 0_{3 \times 3N}\\
        0_{3N\times 3M} & 0_{3N\times 3} & \mathbf{J}
    \end{pmatrix},\label{hess_dyn_leader}
\end{align}
where $\textit{Hess}\bar V^r(\bar x^r)$ is evaluated at $\bar x^r \in \bar \Upsilon^r_0 \setminus \bar{\mathcal{A}}^r_0$. From equation \eqref{grad_dyn_leader}, it follows that the set of critical points of $\bar V^r$ coincides with the equilibrium set $\bar{\Upsilon}_0^r$. Additionally, from equation \eqref{hess_dyn_leader}, we conclude that all critical points of $\bar V^r(\bar x^r)$ in the set $\bar{\Upsilon}_0^r \setminus \bar{\mathcal{A}}^r_0$ are saddle points of $\bar V^r(\bar x^r)$. We now proceed with the proof of the instability of the set $\bar{\Upsilon}_0^r \setminus \bar{\mathcal{A}}^r_0$. Define the real-valued function \( \bar V^r_*(\bar x^r): \bar{\mathcal{S}}^r \rightarrow \mathbb{R} \) such that $\bar V^r_*(\bar x^r) = 2\,\alpha\, k_R \sum_{n \in \mathcal{M}^\pi} (\lambda_{p_n} + \lambda_{d_n}) + 2\,\beta \,k_1 (\lambda_{p} + \lambda_{d})- \bar V^r(\bar x^r)$, where \( \lambda_{p} \) and \( \lambda_{d} \) are two distinct eigenvalues of the matrix \( \bar{A} \), \ie, \( p, d \in \{1, 2, 3\} \) with \( p \neq d \), and the values of \( \alpha \) and \( \beta \) are defined as follows:
\begin{itemize}
    \item \( \alpha = \beta = 1 \) if \( |\mathcal{M}^\pi| > 0 \) and \( \tilde{R}_1 = \mathcal{R}(\pi, \bar{u}_l) \),
    \item \( \alpha = 1 \) and \( \beta = 0 \) if \( |\mathcal{M}^\pi| > 0 \) and \( \tilde{R}_1 = I_3 \),
    \item \( \alpha = 0 \) and \( \beta = 1 \) if \( |\mathcal{M}^\pi| = 0 \) and \( \tilde{R}_1 = \mathcal{R}(\pi, \bar{u}_l) \),
\end{itemize}
 where \( l \in \{1, 2, 3\} \) with \( l \neq p \) and \( l \neq d \). Let $\bar{x}^r_*=(\bar{R}^*_1, \bar{R}^*_2, \dots, \bar{R}^*_M, \tilde R^*_1, \omega^*_1, \omega^*_2, \dots, \omega^*_N) \in \bar \Upsilon^r_0 \setminus \bar{\mathcal{A}}^r_0$ with $\bar{R}^*_n = \mathcal{R}(\pi, u_{l_n})$, $\tilde{R}^*_1 = \mathcal{R}(\pi, u_{l})$ and $\omega^*_i=0$, $\forall n \in \mathcal{M}^\pi$, where $l_n,\, l \in \{1, 2, 3\}$ satisfy $l_n \neq p_n$, $l_n \neq d_n$, $l \neq p$ and $l \neq d$. One can verify $\bar{V}^r_*(\bar{x}^r_*) = 0$. Consider the set $\bar{\bar{\mathbb{B}}}_{\bar{\bar{r}}} := \{ (\bar{R}_1, \bar{R}_2, \dots, \bar{R}_M, \tilde R_1, \omega_1, \omega_2, \dots, \omega_N) \in \bar{\mathcal{S}}^r : |\bar{R}_1^\top \bar{R}_1^*|_I + |\bar{R}_2^\top \bar{R}_2^*|_I + \dots + |\bar{R}_M^\top \bar{R}_M^*|_I+ |\tilde{R}_1^\top \tilde{R}_1^*|_I+||\omega_1-\omega^*_1||+||\omega_2-\omega^*_2||+\dots+||\omega_N-\omega^*_N|| \leq \bar{\bar{r}}\}$ with $\bar{\bar{r}}> 0$. Since the set $\bar \Upsilon^r_0 \setminus \bar{\mathcal{A}}^r_0$ consists only the saddle points of $\bar{V}^r(\bar x^r)$, it follows that the set $\bar{\bar{\mathbb{U}}} = \{ \bar x^r \in \bar{\bar{\mathbb{B}}}_{\bar{\bar{r}}} \mid \bar{V}^r_*(\bar x^r) > 0 \}$ is non-empty. Moreover, one can verify that $\dot{\bar{ V}}^r_*(\bar x^r) = -\dot{\bar{V}}^r(\bar x^r) > 0$ in $\bar{\bar{\mathbb{U}}}$, which implies that any trajectory starting in the set $\bar{\bar{\mathbb{U}}}$ must exit $\bar{\bar{\mathbb{U}}}$ from boundary surface of $\bar{\bar{\mathbb{B}}}_{\bar{\bar{r}}}$. Applying \textit{Chetaev's theorem}, one can show that all points in the undesired equilibrium set $\bar \Upsilon^r_0 \setminus \bar{\mathcal{A}}^r_0$ are unstable. 

To prove item \eqref{dyn_stability_of_equilibrium_leader}, similar to the proof of Theorem \ref{theorem_2}, we will derive the Jacobian matrix corresponding to the dynamics in \eqref{w_dynamics_i}, \eqref{R_bar_dynamics_k}, and \eqref{R1_tilde} under the feedback control law \eqref{torque_i_leader}. Using the first-order approximations from the proofs of items \eqref{dyn_stability_of_equilibrium} and \eqref{stability_of_equilibrium_leader} in Theorems \ref{theorem_2} and \ref{theorem_4}, respectively, one derives the following linearized system
 \begin{align}\label{linear_sys_und_leader}
\begin{bmatrix}
    \dot{\bar r} \\[4pt] \dot{\tilde r}_1\\ \dot y
\end{bmatrix} =
\bar{\bar{\mathcal{J}}}
\begin{bmatrix}
    \bar r \\[4pt] \tilde r_1\\ y
\end{bmatrix},
\end{align}
where $$\bar{\bar{\mathcal{J}}}:=\begin{pmatrix}
        0_{3M\times 3M} & 0_{3M\times 3} & \bar{\bar{\mathcal{J}}}_{13}\\
        0_{3 \times 3M} & 0_{3\times 3} & \bar{\bar{\mathcal{J}}}_{23}\\
         \bar{\bar{\mathcal{J}}}_{31} &  \bar{\bar{\mathcal{J}}}_{32}& \bar{\bar{\mathcal{J}}}_{33} 
    \end{pmatrix},$$
with $\bar{\bar{\mathcal{J}}}_{13}=-\mathbf{H}_c^\top\, \mathbf{R}_c$, $\bar{\bar{\mathcal{J}}}_{23}=\left[R^*_1~~0_{3\times3(N-1)}\right]$, $\bar{\bar{\mathcal{J}}}_{31}=\frac{k_R}{2} \mathbf{J}^{-1} \mathbf{R}_c\,\mathbf{H}_c\, \mathbf{A}^* $, $\bar{\bar{\mathcal{J}}}_{32}=\frac{-k_R}{2} \mathbf{J}^{-1} \begin{bmatrix}
             R^*_1\, \tilde{A}_1 \\ 0_{3(N-1)\times3} 
         \end{bmatrix}$, and $\bar{\bar{\mathcal{J}}}_{33}=-\mathbf{J}^{-1} \left(k_\omega I_{3N}+\bar k_\omega (\mathcal{L}\otimes I_3)\right)$. 
The matrix $\bar{\bar{\mathcal{J}}}$ is the Jacobian matrix corresponding to the dynamics in \eqref{w_dynamics_i}, \eqref{R_bar_dynamics_k}, and \eqref{R1_tilde} under the feedback control law \eqref{torque_i_leader}. Now, we proceed with a proof by contradiction to demonstrate that the matrix $\bar{\bar{\mathcal{J}}}$ does not have any eigenvalues at the origin or purely imaginary. We begin by considering the case of a zero eigenvalue. Assume that $\bar{\bar{\mathcal{J}}}$ has an eigenvalue at the origin, \ie, 
\begin{align}\label{jac_zero}
    \bar{\bar{\mathcal{J}}}\,\begin{bmatrix}
        \bar z_1\\
        \bar z_2\\
        \bar z_3
    \end{bmatrix}=0,
\end{align}
where $\bar z_1 \in \mathbb{R}^{3M}$, $\bar z_2 \in \mathbb{R}^{3}$, and $\bar z_3 \in \mathbb{R}^{3N}$ with $\left[\bar z_1^\top\,\,\bar z_2^\top\,\, \bar z_3^\top\right]^\top \neq 0$. Equation \eqref{jac_zero} implies that 
\begin{align}
    \bar{\bar{\mathcal{J}}}_{13}\, \bar z_3 &= 0 \label{jac_zero_1}\\
    \bar{\bar{\mathcal{J}}}_{23}\, \bar z_3 &= 0 \label{jac_zero_2}\\
    \bar{\bar{\mathcal{J}}}_{31}\, \bar z_1 + \bar{\bar{\mathcal{J}}}_{32}\, \bar z_2 + \bar{\bar{\mathcal{J}}}_{33}\, \bar z_3 &= 0.\label{jac_zero_3}
\end{align}
From \eqref{jac_zero_2}, we conclude that $\bar z_3^1 = 0$, where $\bar z_3^1 \in \mathbb{R}^3$ represents the first three entries of the vector $\bar z_3$. Let \(\mathbf{H}_c=\left[(\mathbf{H}_c^1)^\top \;\; (\mathbf{H}_c^2)^\top\right]^\top\), where
\(\mathbf{H}_c^1 \in \mathbb{R}^{3\times 3M}\) and
\(\mathbf{H}_c^2 \in \mathbb{R}^{3(N-1)\times 3M}\). Following similar steps as in \cite{li2026}, it can be shown that the matrix $\mathbf{H}_c^2$ is nonsingular. Given that the matrix $\mathbf{R}_c$ contains only rotation matrices on its diagonal, it follows from \eqref{jac_zero_1} that $\bar z_3 = 0$. From this and Equation \eqref{jac_zero_3}, one has 
\begin{align}
    -\frac{k_1}{2} \tilde{A}_1\, \bar z_2 + \frac{k_R}{2} \mathbf{H}_c^1 \, \mathbf{A}^*\, \bar z_1 &=0\\
    \frac{k_R}{2}\mathbf{H}_c^2\, \mathbf{A}^* \, \bar z_2 &=0.
\end{align}
Since the matrices $\mathbf{H}_c^2$, $\mathbf{A}^*$, and $\tilde{A}_1$ are nonsingular, it follows that $\bar z_1 = 0$ and $\bar z_2 = 0$. This leads to a contradiction, as eigenvectors cannot be zero by definition. Therefore, the Jacobian matrix $\bar{\bar{\mathcal{J}}}$ does not have an eigenvalue at the origin. Next, we assume that the matrix $\bar{\bar{\mathcal{J}}}$ has purely imaginary eigenvalues, \ie,
\begin{align}
    \bar{\bar{\mathcal{J}}}_{13}\, \bar z_3 &= i\,\lambda\, \bar z_1 \label{jac_imag_1}\\
    \bar{\bar{\mathcal{J}}}_{23}\, \bar z_3 &= i\,\lambda\, \bar z_2 \label{jac_imag_2}\\
    \bar{\bar{\mathcal{J}}}_{31}\, \bar z_1 + \bar{\bar{\mathcal{J}}}_{32}\, \bar z_2 + \bar{\bar{\mathcal{J}}}_{33}\, \bar z_3 &= i\,\lambda\, \bar z_3,\label{jac_imag_3}
\end{align}
with $\left[\bar z_1^\top\,\,\bar z_2^\top\,\, \bar z_3^\top\right]^\top \neq 0$ and $\lambda \neq 0$. After straightforward manipulation of \eqref{jac_imag_1}–\eqref{jac_imag_3}, we obtain the following:
\begin{align}
    -\frac{i}{\lambda}\bar{\bar{\mathcal{J}}}_{13}\, \bar z_3 &= \bar z_1 \nonumber\\
    -\frac{i}{\lambda}\bar{\bar{\mathcal{J}}}_{23}\, \bar z_3 &= \bar z_2 \nonumber\\
     -i\left(\frac{1}{\lambda}\bar{\bar{\mathcal{J}}}_{31}\, \bar{\bar{\mathcal{J}}}_{13} -\frac{1}{\lambda} \bar{\bar{\mathcal{J}}}_{32}\, \bar{\bar{\mathcal{J}}}_{23}- \lambda\,I_{3N}\right) \bar z_3 + \bar{\bar{\mathcal{J}}}_{33}\, \bar z_3 &=0,\label{jac_imag_32}
\end{align}
Now, by multiplying \eqref{jac_imag_32} from the left by $\bar z_3 \mathbf{J}^{-1}$ and considering only the real part of the resulting equation, we get:
\begin{align}
    \bar z_3\, \mathbf{J}^{-1}\,\bar{\bar{\mathcal{J}}}_{33}\, \bar z_3 = k_\omega ||\bar z_3||^2 +\bar k_\omega ||(H^\top\otimes I_3) \bar z_3||^2 =0,
\end{align}
This implies that $\left[\bar z_1^\top ,, \bar z_2^\top ,, \bar z_3^\top\right]^\top = 0$, which contradicts the assumption that the eigenvectors are non-zero. Therefore, the Jacobian matrix $\bar{\bar{\mathcal{J}}}$ does not have any eigenvalues on the imaginary axis. However, all points in the undesired equilibrium set $\bar \Upsilon^r_0 \setminus \bar{\mathcal{A}}^r_0$ are unstable. Thus, the matrix $\bar{\bar{\mathcal{J}}}$ must have at least one eigenvalue with a positive real part. Moreover, by the \textit{stable manifold theorem},
the stable manifold associated with the undesired equilibrium set $\bar \Upsilon^r_0 \setminus \bar{\mathcal{A}}^r_0$ has zero Lebesgue measure. Consequently, the equilibrium set $\bar{\mathcal{A}}^r_0$ is almost globally asymptotically stable. This completes the proof of item \eqref{dyn_stability_of_equilibrium_leader}.



\bibliographystyle{IEEEtran}
\bibliography{References}

\end{document}